\documentclass[12pt,onecolumn,twoside]{IEEEtran}
\usepackage{mathpazo}
\usepackage{times}

\usepackage{graphicx}
\usepackage{amssymb}
\usepackage{amsmath}
\usepackage[mathscr]{euscript}
\usepackage{enumerate}
\usepackage{mathcomp}
\usepackage{mathrsfs}

\usepackage{theorem}  
\usepackage{cite}     
 
\usepackage{upref}

\usepackage{comment}  

\title{Recovery Sets of Subspaces from a Simplex Code\vspace{-1.0ex}}

\date{\today}
\author{\textbf{Yeow Meng Chee$^\text{x}$}, \textbf{Tuvi Etzion$^*$}, \textbf{Han Mao Kiah$^+$}, \textbf{Hui Zhang$^\text{x}$}\\
{\small $^\text{x}$Department of Industrial Systems Engineering and Management, National University of Singapore, Singapore}\\
{\small $^*$Computer Science Department, Technion, Israel Institute of Technology, Haifa 3200003, Israel}\\
{\small $^+$School of Physical and Mathematical Sciences, Nanyang Technological University, Singapore}\\
{\small {\it ymchee@nus.edu.sg}, {\it etzion@cs.technion.ac.il}, {\it hmkiah@ntu.edu.sg}, {\it hzhang.chn@gmail.com}\vspace{-0.13ex}}
\thanks{Parts of this work have been presented at the \emph{IEEE International Symposium on Information Theory 2020}, Los Angeles, California, U.S.A., June 2020 (virtual conference).}
}

\theoremstyle{plain}
\theorembodyfont{\normalfont\slshape}

\newtheorem{thm}{Theorem\hspace{-1pt}}
\newenvironment{theorem}
{\begin{thm}\hspace*{-1ex}{\bf.}}{\end{thm}}

\newtheorem{lem}[thm]{Lemma\hspace{-.75pt}}
\newenvironment{lemma}{\begin{lem}\hspace*{-1ex}{\bf.}}{\end{lem}}

\newtheorem{prop}[thm]{Proposition$\!$}

\newtheorem{cor}[thm]{Corollary$\!$}
\newenvironment{corollary}{\begin{cor}\hspace*{-1ex}{\bf.}}{\end{cor}}

\newtheorem{defn}{Definition$\!$}

\newtheorem{exmp}{Example$\!$}
\newenvironment{example}{\begin{exmp}\hspace*{-1ex}{\bf.}}{\end{exmp}}

\makeatletter
 \DeclareRobustCommand{\nsbinom}{\genfrac[]\z@{}}
 \makeatother
 \newcommand{\sbinom}[2]{\nsbinom{{#1}}{{#2}}}
 \newcommand{\sbinomq}[2]{\nsbinom{{#1}}{{#2}}_{q}}

\newcommand{\cC}{{\cal C}}

\newcommand{\cL}{{\cal L}}

\newcommand{\cV}{{\cal V}}

\DeclareMathAlphabet{\mathbfsl}{OT1}{ppl}{b}{it} 

\newcommand{\field}[1]{\mathbb{#1}}

\newcommand{\F}{\field{F}}
\newcommand{\dP}{\field{P}}

\DeclareMathAlphabet{\mathbfsl}{OT1}{ppl}{b}{it}

\newcommand{\bfzero}{{\bf 0}}

\newcommand{\abs}[1]{\left|#1\right|}

\newcommand{\be}[1]{\begin{equation}\label{#1}}
\newcommand{\ee}{\end{equation}}

\renewcommand{\leq}{\leqslant}
\renewcommand{\ge}{\geqslant}
\renewcommand{\geq}{\geqslant}

\renewcommand{\Bbb}{\mathbb}
\newcommand{\C}{{\Bbb C}}

\newcommand{\mmod}{{\mbox{mod}}}

\newcommand{\Fn}{\smash{\Bbb{F}_{\kern-1pt2}^{\hspace{0.5pt}n}}}

\newcommand{\Fkap}{\smash{\Bbb{F}_{\kern-1pt2}^{\hspace{0.5pt}\kappa}}}

\newcommand{\Cref}[1]{Co\-ro\-lla\-ry\,\ref{#1}}

\newcommand{\Span}[1]{{\left\langle {#1} \right\rangle}}

\makeatletter
\let\over\@@over
\let\atop\@@atop
\let\atopwithdelims\@@atopwithdelims
\makeatother

\makeatletter
\renewcommand{\@endtheorem}{\endtrivlist}
\makeatother

\makeatletter
\renewcommand{\thefigure}{{\bf \@arabic\c@figure}}
\renewcommand{\fnum@figure}{{\bf Figure}\,\thefigure}
\makeatother

\begin{document}

\maketitle

\vspace*{-3.50ex}

\hspace*{-1pt}
\begin{abstract}
Recovery sets for vectors and subspaces are important in the construction
of distributed storage system codes. These concepts are also interesting
in their own right. In this paper, we consider the following very
basic recovery question: what is the maximum number of possible pairwise disjoint recovery sets
for each recovered element? The recovered elements in this work are
$d$-dimensional subspaces of a $k$-dimensional vector space over $\F_q$. Each server stores one representative
for each distinct one-dimensional subspace of the $k$-dimensional vector space, or equivalently
a distinct point of PG$(k-1,q)$. As column vectors, the associated vectors of the stored one-dimensional subspaces
form the generator matrix of the $[(q^k -1)/(q-1),k,q^{k-1}]$ simplex code
over $\F_q$. Lower bounds and upper bounds on the maximum number of such recovery sets are provided.
It is shown that generally, these bounds are either tight or very close to being tight.
\end{abstract}

\vspace{-0.11cm}

\section{Introduction and Preliminaries}
\label{sec:intro}

An important problem in distributed storage systems is the design of a code which can recover any $d$-dimensional
subspace $U$ ($d$-subspace in short) of a $k$-space over $\F_q$ from $n$ disjoint subsets of servers,
where each server stores one distinct 1-subspace of the $k$-space.
Each such subset of servers is called a {\bf \emph{recovery set}} if the subspace spanned by their 1-subspaces
(or equivalently by the vectors of their 1-subspaces) contains $U$.
The set of vectors in this recovery set of servers will be also called a recovery set.
Similarly, the set of 1-subspaces of this recovery set of servers will be also called a recovery set.
Given $d$ and~$k$, one wishes to know what is the minimum number of servers required
for a given multiple recovery of each $d$-subspace of the $k$-space over $\F_q$,
using linear combinations of pairwise disjoint sets of servers.
This problem is of interest for distributed storage system codes and it is called availability~\cite{RPDV16}.
A related problem is to find the availability of distributed storage codes whose codewords are subspaces of possibly higher
dimension~\cite{SES19}. It is also related to a problem associated with a new model in private information
retrieval codes defined for minimizing storage, e.g.~\cite{AsYa18,BlEt19,CKYZ19,FVY15,FVY15a,KuYa21,NaYa22,VaYa23,YoYa22,ZEY19}.
In general, coding for subspaces has become quite fashionable, mainly for wide applications for distributed storage
and also for random network coding, e.g.~\cite{BEOVW,EtSt16,KoKu08,KoKs08,EtVa11,RaEt15}.

In this paper, we will solve another related problem for distributed storage systems. Let $d$ and $k$ be integers such that
$1 \leq d \leq k$, let $U$ be any $d$-subspace of $\F_q^k$, and assume that each server stores a 1-subspace of~$\F_q^k$.
What is the maximum number of subsets of pairwise disjoint servers
such that each such subset can recover~$U$, i.e., their 1-subspaces span $U$?
This is equivalent to the recovery of $d$-subspaces from the columns of the generator matrix
of the simplex code. Recovering a batch of elements from the columns of a generator matrix for the $[(q^k -1)/(q-1),k,q^{k-1}]$ simplex code
was considered in the past~\cite{HKRS,WKCB17,YoYa22,ZEY19}, but only for binary alphabet and not
for subspaces. This work is a generalization in this direction.

Let $\F_q^k$ be the vector space of dimension $k$ over $\F_q$,
the finite field with $q$ elements, $q \in \dP$, where $\dP$~is the set of prime powers.
Clearly, $\F_q^k$ consists of $q^k$ vectors of length $k$ and it contains $\frac{q^k-1}{q-1}$ distinct
1-subspaces. In general $\F_q^k$ contains $\sbinom{k}{\ell}$ distinct $\ell$-subspaces, where
\begin{small}
$$
\sbinomq{k}{\ell}
\triangleq \frac{(q^k-1)(q^{k-1}-1) \cdots
(q^{k-\ell+1}-1)}{(q^\ell-1)(q^{\ell-1}-1) \cdots (q-1)}
$$
\end{small}
is the well-known \emph{$q$-binomial coefficient}, known also as the \emph{Gaussian coefficient}.
To distinguish between the set of $q^k$ vectors of $\F_q^k$
and its 1-subspaces we will denote the set of all 1-subspaces
of~$\F_q^k$ by $V_q^k$. When $q=2$, $V_q^k$ corresponds to the nonzero elements of $\F_q^k$ and hence we will not
distinguish between them. We note that each 1-subspace of $V_q^k$ can be represented by $q-1$ elements of~$\F_q^k$.
Two such representations can be obtained from each other using multiplication by an element of~$\F_q$.
Such representation for a 1-subspace of $V_q^k$ can be viewed as a projective point in the
projective geometry PG($k-1,q$). Note, that a $d$-subspace in $\F_q^k$ is a $(d-1)$-subspace in PG$(k-1,q)$.
Using this representation the generator matrix of the $[(q^k -1)/(q-1),k,q^{k-1}]$ simplex code contains
the projective points of the projective geometry PG($k-1,q$).
Finally, we will use the isomorphism between the finite field $\F_{q^n}$ in a way that sometimes we will consider
the information in the servers as elements in the vector space $\F_q^n$ and
sometimes the information in the servers will be presented by elements from the finite field $\F_{q^n}$.
But, there will never be a mixed notation in the representation.

In a distributed storage system, a server stores an $\ell$-subspace of~$\F_q^k$ and a typical user
is interested in a $d$-subspace $U$ of~$\F_q^k$. The information of this $d$-subspace is usually stored
in sub-packets across several servers. It is quite common that some servers, which hold important sub-packets
of $U$, are not available. In this case, it is required that it will be possible to recover~$U$ from other sets
of servers. Hence, it is required that for some~$n$, each $d$-subspace will be recovered
from $n$ pairwise disjoint sets of servers.

Assume now, that each server stores a different 1-subspace of $\F_q^k$, i.e., an element of $V_q^k$.
Let $N_q(k,d)$ be the maximum number of such pairwise disjoint sets of servers that will be able to recover any given $d$-subspace $U$.
W.l.o.g. (without loss of generality), in some places, we will assume that $U$ is $\F_q^d$ (in $\F_q^k$) since
$\F_q^k$ can be written as $W + U$ for any $d$-subspace $U \in \F_q^k$ and some $(k-d)$-subspace~$W$ of $\F_q^k$.
The $(k-d)$-subspace $W$ can be taken w.l.o.g. and abuse of notation as $\F_q^{k-d}$. In this case, an element of $\F_q^k$
can be written as $(x,y)$, where $x \in W$ and $y \in U$. It can be also written as $x+y$. Before we start to discuss the
general results in the following sections, let us consider a simple scenario that could be very useful in the sequel.
\begin{theorem}
\label{thm:basic_parts}
The set $V_q^d$ contains $n=\left\lfloor \frac{q^d -1}{d(q-1)} \right\rfloor$ disjoint recovery sets for the $d$-subspace $V_q^d$.
\end{theorem}
\begin{proof}
Let $\alpha$ be a primitive element of $\F_{q^d}$. Any $d$ consecutive powers of $\alpha$ are linearly independent
and any $\frac{q^d-1}{q-1}$ consecutive powers of $\alpha$ are contained in a distinct set of 1-subspaces.
Hence, if we set $S_i = \{\alpha^{(i-1)d},\alpha^{(i-1)d+1},\ldots, \alpha^{id-1}\}$,
we have that $\langle S_i\rangle = \F_q^d$ for any ${1 \leq i \leq n}$.
\end{proof}

Since any two $d$-subspaces (possibly of two different spaces with different dimensions) are isomorphic we have the following consequence.
\begin{corollary}
\label{cor:basic_parts}
If $U$ is the $d$-subspace to be recovered from $U$, then the 1-subspaces of~~$U$ can be partitioned into
$\left\lfloor \frac{q^d -1}{d(q-1)} \right\rfloor$ recovery sets for $U$ and possibly one more subset with less than $d$ 1-subspaces.
\end{corollary}
\begin{corollary}
For each $q \geq 2$ and $k \geq 2$ we have
$$
N_q (k,k) = \left\lfloor \frac{q^k-1}{k(q-1)} \right\rfloor~.
$$
\end{corollary}

We start by stating the most simple upper bound and the most simple lower bound on $N_q(k,d)$.
First, the upper bound will be derived.
\begin{theorem}
\label{thm:gen_upper}
If $q \in \dP$ and $k \geq d$ is a positive integer, then
$$
N_q (k,d) \leq \left\lfloor \frac{q^d-1}{d(q-1)} \right\rfloor + \left\lfloor \frac{\ell (q-1) + q^k - q^d}{(d+1)(q-1)} \right\rfloor ,
$$
where $\ell$ is the reminder from the division of $\frac{q^d -1}{q-1}$ by $d$.
\end{theorem}
\begin{proof}
Let $U$ be the $d$-subspace which has to be recovered.
Recovery sets of size $d$ can be obtained only from $d$ linearly independent 1-subspaces of $U$. By Theorem~\ref{thm:basic_parts},
there are at most $\left\lfloor \frac{q^d-1}{d(q-1)} \right\rfloor$ such recovery sets.
The number of remaining nonzero elements from $U$ is $\ell (q-1)$ and the number of nonzero elements in $\F_q^k$ which
are not contained in $U$ is $q^k - q^d$.
All the remaining recovery sets are from these elements and they must be of size at least $d+1$ which implies the claim of the theorem.
\end{proof}
The bound of Theorem~\ref{thm:gen_upper} will be improved when we are more precise in the number of recovery sets
of size $d+1$ and we will have some approximation on the number of recovery sets of larger size.

We continue to derive a related simple lower bound. For this bound, we need the following lemma and its consequence.
\begin{lemma}
\label{lem:coeff_sum_not0}
If $f(z)= z^d + a_{d-1}z^{d-1} + \cdots + a_1z + a_0$ is a primitive polynomial over $\F_q$, then
$$
\sum_{i=0}^{d-1} a_i  +1 \neq 0.
$$
\end{lemma}
\begin{proof}
If $\sum_{i=0}^{d-1} a_i  +1 =0$, then $f(1)=1+\sum_{i=0}^{d-1} a_i =0$ and hence 1 is a root of $f(z)$, a contradiction.
\end{proof}

For $x_1,x_2,\ldots, x_\ell \in \F_{q^n}$ and $\gamma \in \F_q$ let $\gamma (x_1,x_2,\ldots,x_\ell)=(\gamma x_1,\gamma x_2,\ldots,\gamma x_\ell)$.

\begin{corollary}
\label{cor:primitiveRecover}
If $\alpha$ is a root of the primitive polynomial $f(z)= z^d + a_{d-1}z^{d-1} + \cdots + a_1z + a_0$ over~$\F_q$, then the
$d$ vectors in
$$
R \triangleq \{(x,\alpha^i),(x,\alpha^{i+1}),\ldots,(x,\alpha^{i+d}) \}
$$
form a recovery set for $\F_q^d$ in $\F_q^k$, where $x \in \F_{q^{k-d}} \setminus \{ 0 \}$. (note that $\alpha$ is a primitive element in $\F_{q^d}$.).
In other words, $R$ is a recovery set for $\F_q^d$.
\end{corollary}
\begin{proof}
Since $\alpha$ is a root of $f(z)$, it follows that
$$
\alpha^d + a_{d-1}\alpha^{d-1} + \cdots + a_1 \alpha + a_0 =0.
$$
and hence
$$
\alpha^{i+d} + a_{d-1}\alpha^{i+d-1} + \cdots + a_1 \alpha^{i+1} + a_0 \alpha^i =0, ~~~ \text{for~each} ~ i \geq 0.
$$
On the other hand, since by Lemma~\ref{lem:coeff_sum_not0}, we have that $\sum_{j=0}^{d-1} a_j  +1 \neq 0$, it follows that
$$
(x,\alpha^{i+d}) + a_{d-1}(x,\alpha^{i+d-1}) + \cdots + a_1 (x,\alpha^{i+1}) + a_0 (x,\alpha^i)
$$
$$
= (x,\alpha^{i+d}) + \sum_{j=0}^{d-1} a_j(x,\alpha^{i+j}) = ( \left(\sum_{i=0}^{d-1} a_i  +1 \right) x, \bfzero)  =(\gamma x, \bfzero) ,
$$
for some $\gamma \in \F_q \setminus \{ 0 \}$.
Thus, $(x,\bfzero)$ can be recovered from the vectors in $R$ and therefore
\begin{equation}
\label{eq:d_independent}
(\bfzero,\alpha^i), (\bfzero,\alpha^{i+1}), \ldots , (\bfzero,\alpha^{i+d-1}), (\bfzero,\alpha^{i+d})
\end{equation}
can be recovered too. Since each $d$ consecutive elements in Eq.~(\ref{eq:d_independent}) are linearly independent, it follows that
$$
\{ (x,\alpha^i),(x,\alpha^{i+1}),\ldots,(x,\alpha^{i+d}) \}
$$
is a recovery set for $\F_q^d$ in $\F_q^k$.
\end{proof}

The last part of the proof of Corollary~\ref{cor:primitiveRecover} can be used to prove the following lemma.
\begin{lemma}
\label{lem:simple_recover}
If $\alpha$ is a primitive element in $\F_{q^d}$ and $x \in \F_{q^{k-d}} \setminus \{ 0 \}$, then
$$
R \triangleq \{(x,\bfzero),(x,\alpha^i),(x,\alpha^{i+1}),\ldots,(x,\alpha^{i+d-1}) \}
$$
is a recovery set for $\F_q^d$ in $\F_q^k$, where $x \in \F_{q^{k-d}} \setminus \{ 0\}$.
\end{lemma}

\begin{theorem}
\label{thm:tight}
If $q \in \dP$ and $k \geq d$ is a positive integer, then
\begin{equation}
\label{eq:bound1}
N_q (k,d) \geq \left\lfloor \frac{q^d-1}{d(q-1)} \right\rfloor + \left\lfloor \frac{q^d}{d+1} \right\rfloor \frac{q^{k-d}-1}{q-1}.
\end{equation}
Furthermore, if $d+1$ divides $q^d$, then there is equality in Eq.~(\ref{eq:bound1}).
\end{theorem}
\begin{proof}
The lower bound is obtained similarly to the upper bound of Theorem~\ref{thm:gen_upper} by observing that recovery sets
of size $d$ can be obtained only from the elements of $U$. From the remaining space, we construct as many as possible recovery sets of size $d+1$.
Let $\cL$ be the set of nonzero vectors from~$\F_q^{k-d}$ whose first nonzero entry is a \emph{one}. Clearly,
$$
\abs{\cL} = q^{k-d-1} + q^{k-d-2} + \cdots + q +1 = \frac{q^{k-d}-1}{q-1},
$$
and the set $\{ (x,y) ~:~ x \in \cL , ~ y \in \F_q^d \} \cup \hat{U}$, where $\hat{U}$ is the set of one
subspaces of $U$, is a set of representatives for $V_q^k$.
By Corollary~\ref{cor:primitiveRecover} and Lemma~\ref{lem:simple_recover},
for each $x \in \cL$ we can form $\left\lfloor \frac{q^d}{d+1} \right\rfloor$ recovery sets for $U$.

From $U$ we obtain $\left\lfloor \frac{q^d-1}{d(q-1)} \right\rfloor$ recovery sets. There are $\frac{q^{k-d}-1}{q-1}$
elements in $\cL$ and each one yields $\left\lfloor \frac{q^d}{d+1} \right\rfloor$ recovery sets.
Thus, the claim of the theorem follows.
\end{proof}

The value of $N_q(k,1)$ was analyzed in~\cite{CEKZ20}, where the following results were obtained.

\begin{theorem}
\label{thm:one_dim}
Let $q \in \dP$ and $n \geq 2$ be positive integers.
\begin{itemize}
\item[(1)] If $q$ is even, then $N_q(k,1) = 1 + \frac{q^k - q}{2(q-1)}$.

\item[(2)] If $q$ is odd, then $N_q(k,1)= 1 + \frac{q^{k-1}-1}{2} + \left\lfloor \frac{q^{k-1}-1}{3(q-1)} \right\rfloor$.
\end{itemize}
\end{theorem}

In this paper, an analysis for $d > 1$ will be done and the rest of this paper is organized as follows.
In Section~\ref{sec:main_ideas} we describe some
of the ideas for our main construction for a lower bound on the maximum number of recovery sets.
We demonstrate these ideas first only for $q=2$,
where we don't distinguish between the vectors of $\F_2^n$ and its 1-subspaces.
The ideas for a lower bound will be demonstrated for $d=2^m-1$ and subspaces over $\F_2$ using a binary perfect code
of length $2^m-1$. Also, some ideas for an upper bound will be demonstrated for $d=2$ using integer programming.
Tight bounds or almost tight bounds
for $q=2$ and specific small values of $d$, will be presented in Sections~\ref{sec:lower_binary} and~\ref{sec:upper_binary},
where lower bounds will be considered in the first and upper bounds will be considered in the second.
The bounds on the number of recovery sets for $q >2$ are presented in Section~\ref{sec:general}. Some of the bounds that we obtain
are tight, e.g., when $k-d$ is even and $d+2$ divides~${q+1}$. In general, optimality can be obtained
in other cases by constructing certain partitions.
Furthermore, the upper bound for $q>2$ will be improved compared to $q=2$.
Conclusion and problems for further research are presented in Section~\ref{sec:conclusion}.

\section{The Main Ideas for a Lower Bound and an Upper Bound}
\label{sec:main_ideas}
\vspace{0.0cm}

We start this section by describing first our main ideas for constructing many recovery sets from the
columns of the parity check matrix of the simplex code over $\F_q$. This will yield a lower bound on
the maximum number of possible recovery sets. We will continue by describing the main ideas to obtain
an upper bound on this number. Then in subsection~\ref{sec:perfect} we will demonstrate the main ideas
to obtain a lower bound by using binary perfect codes. In subsection~\ref{sec:Integer_program} we will
rephrase the main ideas of the upper bound in terms of integer programming to obtain a specific upper bound.
Most of the discussion in this section is for $q=2$, but the main ideas will be similar for $q>2$.

We start by first representing the points of PG$(k-1,2)$ which are the 1-subspaces of $V_2^k$ by
a ${2^{k-d} \times 2^d}$ matrix $T$.
The rows of the matrix $T$ will be indexed and identified by the $2^{k-d}$ elements of $\F_2^{k-d}$
and the columns of $T$ by the elements of $\F_2^d$, where the first one will be $\bf0$, and the other columns
will be identified, and ordered, as $\alpha^0,\alpha^1,\alpha^2,\ldots,\alpha^{2^d-2}$ and also by their appropriate
vector representation of length $d$, where $\alpha$ is a primitive element
in $\F_{2^d}$. Entry $T(x,y)$, ${x \in \F_2^{k-d}}$, ${y \in \F_2^d}$, in the matrix will contain the pair $(x,y)$,
which represents an element in~$\F_2^k$.
It will be equivalent if the entry $T(x,y)$ will contain the element $x+y$, where $x \in W$, $y \in U$, $W$ is
a ${(k-d)\text{-subspace}}$ and $U$ is a $d$-subspace such that $\F_2^k = W + U$.
The elements of the matrix $T$ represent together all the 1-subspaces (elements)
of $\F_2^k$, i.e., the points of PG$(k-1,2)$. Note, that $T(0,0)$ is the only entry of~$T$ which does not contain a 1-subspace.

\begin{lemma}
\label{lem:representPG}
The entries of the matrix $T$, excluding $T(0,0)=0$, represent exactly all the points of PG$(k-1,2)$
(which are the 1-subspaces of $\F_2^k$).
\end{lemma}
%

Using the representation of $\F_2^k$ with the matrix $T$ we can now design recovery sets for $\F_2^d$ as follows.
\begin{enumerate}
\item Any $d$ consecutive powers of $\alpha$ in the first row of $T$ are linearly independent and hence can be used to generate
$\lfloor \frac{2^d-1}{d} \rfloor$ pairwise disjoint recovery sets (see Theorem~\ref{thm:basic_parts}).

\item Any $d+1$ consecutive elements in a row $x$ of $T$, $x \neq \bf0$, span $\F_2^d$ (see Corollary~\ref{cor:primitiveRecover}).
Similarly, any $d$ consecutive elements in a row $x$ of $T$ together with $T(x,\bf0)$ span $\F_2^d$ (see Lemma~\ref{lem:simple_recover}).
Hence, any such row of $T$ can be used to generate $\left\lfloor \frac{2^d}{d+1} \right\rfloor$ pairwise
disjoint recovery sets. The same is true if there exists another construction for partition (or almost a partition)
of the elements in a row of $T$ into sets for which each one spans $\F_2^d$.

\item The remaining elements in each row are combined and recovery sets, usually of size $d+2+\epsilon$, for small $\epsilon$,
are generated.
\end{enumerate}

We will try to improve the bound of Theorem~\ref{thm:tight} in a simple way. After the construction of $\lfloor \frac{q^d}{d+1} \rfloor$
recovery sets from each row of the matrix $T$, not all the elements in~$T$ will be used
in these recovery sets. The remaining elements, which will be called {\bf \emph{leftovers}}, can be used to obtain more
recovery sets whose size is $d+2$ or more (as we will see in the sequel). If there were many leftovers in the first row of $T$, then
some recovery sets of size $d+1$ can be obtained from these leftovers and combined with the leftovers from the other rows of $T$.
For the other leftovers, it will be proved that in some cases recovery sets of size $d+2$ are sufficient.
If $d+1$ divides $q^d$, there are no leftovers in each row and hence we cannot obtain more recovery sets. Note,
that in such case the leftovers from the first row of~$T$ cannot be of any help to provide more recovery sets. Hence,
the bound of Theorem~\ref{thm:tight} in Eq.~(\ref{eq:bound1}) is attained.

The representation of PG$(k-1,2)$ with the matrix $T$ can be used also to obtain an upper bound
on the maximum number of recovery sets. Recovery sets of size $d$ can be constructed only from the elements
in the first row of $T$ and we cannot have more than $\left\lfloor \frac{2^d -1}{d} \right\rfloor$ such recovery sets.
It is readily verified that involving elements from the first row of $T$ with elements from the other rows of $T$ will require more
than $d$ elements in the recovery sets. This will imply that the leftovers from the first row of $T$ can be used with elements from the other
rows of $T$ to form a recovery set of size at least~${d+1}$. Other leftovers in other rows
can be used in recovery sets of size at least $d+2$. This can imply
bounds that are better than the one in Theorem~\ref{thm:gen_upper}. The upper bounds obtained by these arguments can
be found either by careful analysis or by using integer programming.

In the two subsections which follow we will give examples for the implementation of these ideas.
In Sections~\ref{sec:lower_binary}, \ref{sec:upper_binary}, and~\ref{sec:general} we will implement these ideas for various parameters
and also for $q > 2$.

\subsection{Recovery Sets via Perfect Codes}
\label{sec:perfect}

To demonstrate our main ideas for lower bounds in the following sections, we start with a special case which on
one hand explains the main idea, and on the other hand it has its beauty as it is
based on binary perfect codes with radius one, and finally, it is optimal.
For a word $x \in \F_2^n$, the {\bf \emph{ball}} of radius one around~$x$, $B(x)$, is the
set of words at Hamming distance at most one from $x$, i.e.
$$
B(x) \triangleq \{ y ~:~ y \in \F_2^n , ~ d_H (x,y) \leq 1 \},
$$
where, $d_H(x,y)$ is the Hamming distance between $x$ and $y$.
A {\bf \emph{binary perfect code}} $\cC$ of length $n$ is a set of binary words of length $n$,
such that the balls of radius one around the codewords of $\cC$ form a partition
of $\F_2^n$. It is well known~\cite{Etz22} that a binary perfect code
exists for each $n=2^m-1$ and there are many nonequivalent such codes from which only
one is a linear code, {\bf \emph{the Hamming code}}. Assume now that the dimension $d$
for the subspace to be recovered is of the form $d=2^m-1$. Let $U$ be
the $d$-subspace to be recovered and $W$ be a $(k-d)$-subspace of $\F_2^k$ such that $W + U = \F_2^k$.
The vectors stored in the servers are represented by the $2^{k-d} \times 2^d$ array $T$.
By Corollary~\ref{cor:basic_parts},
the first row of $T$ is partitioned into $\left\lfloor \frac{2^d-1}{d} \right\rfloor$ recovery
sets. Now, let $\cC$ be any binary perfect code of length $d=2^m-1$ (linear or nonlinear).
$\F_2^d$ is partitioned into $\frac{2^d}{d+1}$ balls of radius one whose centers are the codewords
of~$\cC$. Consider row $x$ of $T$, where $x$ is nonzero. This row is partitioned into $\frac{2^d}{d+1}$
subsets, where the $i$-th subset is $P_i \triangleq \{ x + y ~:~ y \in B_i \}$ and $B_i$ is the $i$-th ball
from the partition of $\F_2^d$ into balls, of radius one, by the perfect codes $\cC$.
Clearly, $B_i = \{ c_i \} \cup \{ c_i + e_j ~:~ 1 \leq j \leq d \}$,
where $c_i$ is the $i$-th codeword of $\cC$ and $e_j$ is a unit vector with a \emph{one} in the $j$-th coordinate. It implies that
$P_i = \{ x+c_i \} \cup \{ x+c_i + e_j ~:~ 1 \leq j \leq d \}$ and hence $e_j \in \Span{P_i}$ for
each $1 \leq j \leq d$. Therefore, $U \subset \Span{P_i}$ and $P_i$ is a recovery set for $U$.
Since a recovery set of size $d$ can be obtained only from elements of $U$, it follows that except for
the $\left\lfloor \frac{2^d-1}{d} \right\rfloor$ recovery sets obtained from the first row of $T$ all the other recovery sets
must be of size at least $d+1$. Thus, we have a special case of Theorem~\ref{thm:tight}.
\begin{theorem}
\label{thm:recover_perfect}
If $d=2^m-1$ then
$$
N_2(k,d) = \left\lfloor \frac{2^d-1}{d}  \right\rfloor + \frac{2^k-2^d}{d+1}
$$
\end{theorem}

It should be noted that the same result can be obtained by using the construction based on Theorem~\ref{thm:basic_parts}
and Corollary~\ref{cor:primitiveRecover}. We introduced the construction based on perfect codes for its uniqueness.

\subsection{Integer Programming Bound}
\label{sec:Integer_program}
\vspace{0.0cm}

We continue and demonstrate an upper bound on $N_2(k,2)$ using integer programming.
As before, consider the $2^{k-2} \times 4$ array $T$ and seven types of recovery sets from this array,
whose entries, except for $T(0,0)$, are exactly the elements stored in the $2^k -1$ servers. Now, the variables
which denote the number of recovery sets of each type are defined as follows.
\begin{enumerate}
\item The first type has recovery sets with two elements from the first row of $T$. The number
of recovery sets from this type will be denoted by $X_1$.

\item The variable $X_2$ denotes the number of recovery sets with one element from the first row and two elements
from an internal row (not the first one) of $T$. (note that there is no recovery set when one element is taken
from the first row of $T$ and two elements are taken from two different internal rows of $T$.)

\item The variable $X_3$ denotes the number of recovery sets with one element from the first row and three elements
from three internal rows of $T$.

\item The variable $Y_3$ denotes the number of recovery sets with three elements from one of the internal rows
of $T$.

\item The variable $Y_{22}$ will denote the number of recovery sets with two elements from one internal row of $T$ and
two elements from another internal row of $T$.

\item The variable $Y_4$ denotes the number of recovery sets with four elements: two elements from one internal row of $T$ and
two elements from two other different internal rows of~$T$.

\item The variable $Y_5$ denotes the number of recovery sets with five elements from five
different internal rows of $T$.
\end{enumerate}
Note, that four elements from four different internal rows cannot be used to form a recovery set.
It can be verified that each other possible recovery set contains by its definition one
of these seven types and hence can be replaced to save some servers for other recovery sets.

We continue with a set of three inequalities related to these variables. The first row
of $T$ contains three elements and hence we have that
$$
2X_1 + X_2 + X_3 \leq 3 ~.
$$
The second inequality is also very simple, we just sum the contribution of each type to the
total number of elements in the internal rows of $T$ which have a total of $2^k-4$ elements.
This implies that
$$
2X_2 + 3X_3 + 3Y_3 + 4Y_{22} + 4Y_4 + 5Y_5 \leq 2^k -4~.
$$
For the last equation, we consider any specific internal row of~$T$ and the types
that are using at least two elements from this row. The number of related recovery
sets using this row will be denoted by $X'_2,Y'_3,Y'_{22},Y'_4$ and we have the inequality
$$
2X'_2 + 3Y'_3 + 2Y'_{22} + 2 Y'_4 \leq 4 ~.
$$
But, since $Y'_3 \in \{0,1\}$ while the other variables can get the values 0,1, or 2, it follows that
$$
2X'_2 + 4Y'_3 + 2Y'_{22} + 2 Y'_4 \leq 4 ~.
$$
Next, we sum this equation over all the $2^{k-2}-1$
internal rows, taking into account that $Y'_{22}$ in the equation is counted
for two distinct rows. Hence, we have
$$
2X_2 + 4Y_3 + 4Y_{22} + 2Y_4 \leq 2^k -4~.
$$

We can now write the following linear programming problem
$$
\gamma = \text{maximize}~~X_1 + X_2 + X_3 + Y_3 + Y_{22} + Y_4 + Y_5
$$
subject to nonnegative integer variables and the following three constraints
$$
2X_1 + X_2 + X_3 \leq 3
$$
$$
2X_2 + 3X_3 + 3Y_3 + 4Y_{22} + 4Y_4 + 5Y_5 \leq 2^k -4
$$
$$
2X_2 + 4Y_3 + 4Y_{22} + 2Y_4 \leq 2^k -4~.
$$
We apply now the dual linear programming problem~\cite{Chv83}
$$
\gamma = \text{minimize}~~3Z_1 + (2^k-4) Z_2 + (2^k-4)Z_3
$$
subject to nonnegative variables and the following seven constraints
$$
2 Z_1 \geq 1
$$
$$
Z_1 + 2Z_2 + 2Z_3 \geq 1
$$
$$
Z_1 + 3Z_2 \geq 1
$$
$$
3Z_2 + 4Z_3 \geq 1
$$
$$
4Z_2 + 4Z_3 \geq 1
$$
$$
4Z_2 + 2Z_3 \geq 1
$$
$$
5Z_2  \geq 1
$$

Using IBM software CPLEX
it was found that
this optimization problem has exactly one solution, ${Z_1=\frac{1}{2}}$, $Z_2=\frac{1}{5}$, and ${Z_3 = \frac{1}{10}}$,
and hence $\gamma = \frac{3}{2} +\frac{3(2^{k-1}-2)}{5}$. But, since our solution is an integer solution, it follows that
$\gamma = \left\lfloor \frac{3}{2} +\frac{3(2^{k-1}-2)}{5} \right\rfloor$ and as a consequence.

\begin{lemma}
\label{lem:two_upper}
For $k \geq 2$, $N_2(k,2) \leq \left\lfloor \frac{3}{2} +\frac{3(2^{k-1}-2)}{5} \right\rfloor
= \left\lfloor \frac{3 \cdot 2^k +3 }{10}  \right\rfloor$.
\end{lemma}

This method of integer programming can be used for general parameters, i,e., also for $q >2$, and
also for $d>2$. We will omit the evolved computations which are increased as $q$ and $d$ are increased.

\section{Constructions of Recovery Sets for Binary Alphabet}
\label{sec:lower_binary}

In this section lower bounds on the number of recovery sets for binary alphabet, when the subspace to recover
has a low dimension, will be derived.

\subsection{Recovery Sets for Two-Dimensional Subspaces}
\label{sec:small_sets}

The goal in this subsection is first to find the value of $N_2(k,2)$ which considers binary 2-subspaces.
The proof consists of three steps. In the first step, we will
take a recovery set of size two (it is not possible to have two such disjoint recovery sets) and generate as many possible
pairwise disjoint recovery sets of size three which is the best possible. In the second step, we will continue and from
the remaining 1-subspaces (the leftovers) we generate pairwise disjoint recovery sets of size five.
This construction will provide a lower bound on the number of recovery sets. For the upper bound, in the third step,
we will use either integer programming or careful analysis for the possible size of recovery sets to prove that our construction
yields the largest possible number of pairwise disjoint recovery sets.
The integer programming was demonstrated in Section~\ref{sec:Integer_program}.
We can also obtain similar results to the ones obtained with integer programming,
by careful analysis. After that, we will finish to find the value of $N_2(k,2)$.

Consider the $2^{k-2} \times 4$ array $T$ whose rows are indexed by the vectors of $\F_2^{k-2}$, where the first row
is indexed by $\bf 0$. The columns are indexed by $\bf 0$, $u$, $v$, and $u+v$, where
$U=\Span{ u,v}$ is the 2-subspace to be recovered.

We have seen that any two nonzero elements from the first row span the subspace $U$ and any three elements from any other row also span $U$.
Together we have $2^{k-2}$ recovery sets and $2^{k-2}$ leftovers,
one from each row, where these elements can be chosen arbitrarily. We will prove that these
leftovers can be partitioned into recovery sets of size five and possibly one or two sets
of a different small size. This will be based on a structure
that we call a {\bf \emph{quintriple}} (as it consists of five elements
and three ways to generate one of its elements). We start with the definition of a quintriple.

A subset $\{ x_1,x_2,x_3,x_4,x_5 \} \subset \F_2^{k-2}$ is called a {\bf \emph{quintriple}} if
$x_1 = x_2 + x_3 = x_4 + x_5$. A quintriple yields one recovery set with one
element from each row associated with the elements of the quintriple. As the leftovers
in each row can be chosen arbitrarily we choose the following elements in the rows indexed by $x_1,x_2,x_3,x_4$, and $x_5$.
$$
(x_1 , \bfzero )
$$
$$
(x_2 , u)
$$
$$
(x_3 , u+v)
$$
$$
(x_4 , v)
$$
$$
(x_5 , u+v).
$$
Since
$$
(\bfzero,v)= (x_1, \bfzero ) + (x_2,u) + (x_3,u+v)
$$
and
$$
(\bfzero , u )=(x_1, \bfzero ) + (x_4 ,v) + (x_5,u+v) ,
$$
it follows that $U$ is recovered from $\{ (x_1, \bfzero) (x_2,u) ,(x_3,u+v),(x_4,v),(x_5,u+v) \}$.
Thus, our goal is to form as many as possible pairwise disjoint quintriples of $\F_2^{k-2}$,
and possibly one or two more subsets of small size, to complete our construction.

Now, we will consider a partition of $\F_2^m$ into disjoint quintriples
and possibly one or two more subsets of a small size. Each quintriple
will yield one recovery set; one more recovery set will be obtained
from the other elements (outside the quintriples). We distinguish between four cases,
depending on whether $m$ is congruent to 0, 1, 2, or 3 modulo $4$. The construction will be recursive,
where the initial conditions are quintriples for $m=4$, $m=5$, $m=6$, and $m=7$. For the recursive
construction, rank-metric codes and lifted rank-metric codes will be introduced, and their definitions are
provided now.

For two $\tau \times \eta$ matrices $A$ and $B$ over $\F_q$ the \emph{rank distance}, $d_R (A,B)$, is defined by
$$
d_R (A,B) \triangleq \text{rank}(A-B)~.
$$
A code $\cC$ is a $[\tau \times \eta,\ell,\delta]$ rank-metric code if its codewords are $\tau \times \eta$
matrices over $\F_q$, which form a linear subspace of dimension $\ell$ of $\F_q^{\tau \times \eta}$, and
each two distinct codewords $A$ and $B$ satisfy $d_R (A,B) \geq \delta$. Rank-metric codes were well
studied~\cite{Gab85,Rot91,Del78}. A Singleton bound for rank-metric codes was proved in~\cite{Del78,Gab85,Rot91}:

\begin{theorem}
\label{thm:MRD_bound}
If $\cC$ is a $[\tau \times \eta,\ell,\delta]$ rank-metric code, then
\begin{equation}
\label{eq:MRD_bound}
\ell \leq \textup{min}\{\tau(\eta-\delta+1),\eta(\tau-\delta+1)\},
\end{equation}
and this bound is attained for all possible parameters.
\end{theorem}

Codes which attain the upper bound of Eq.~(\ref{eq:MRD_bound})
are called {\bf \emph{maximum rank distance}} codes (or {\bf \emph{MRD}} codes in short).

A $\tau \times \eta$ matrix $A$ over $\F_q$ is {\bf \emph{lifted}} to a $\tau$-subspace of $\F_q^{\tau+\eta}$
whose generator matrix is $[I_\tau~A]$, where $I_\tau$ is the identity matrix of order $\tau$.
A $[\tau \times \eta,\ell,\delta]$ rank-metric code $\cC$
is \emph{lifted} to a code~$\C$ of $\tau$-subspaces in $\F_q^{\tau+\eta}$ by lifting
all the codewords of~$\cC$, i.e., $\C \triangleq \{ \Span{[I_\tau ~ A]}_R ~:~ A \in \cC \}$, where
$\Span{B}_R$ is the linear span of the rows of $B$. This code will be denoted by $\C^{\text{MRD}}$.
If $\delta = \tau$ then all the codewords ($\tau$-subspaces) of $\C$ are pairwise disjoint (intersect in the null-space $\{ {\bf 0}\}$).
These codes have found applications in random network coding, where their constructions
can be found~\cite{EtSi09,EtSi13,EtSt16,KoKs08,SKK08}.

We start by describing the construction of quintriples from~$\F_2^4$ (one can find such quintriples
by an appropriate computer search).
Let $\alpha$ be a root of the primitive polynomial ${x^4+x+1}$, i.e., $\alpha^4 = \alpha +1$.
The nonzero elements of the finite field $\F_{16}$ are given in the following table.

\begin{table}[h]
\centering
\begin{tabular}{|c|c|c|c|c|c|}
  \hline
   & $\alpha^3$ & $\alpha^2$ & $\alpha^1$ & $\alpha^0$  \\ \hline
  $\alpha^0$ & 0 & 0 & 0 & 1  \\\hline
  $\alpha^1$ & 0 & 0 & 1 & 0  \\\hline
  $\alpha^2$ & 0 & 1 & 0 & 0  \\\hline
  $\alpha^3$ & 1 & 0 & 0 & 0  \\\hline
  $\alpha^4$ & 0 & 0 & 1 & 1  \\\hline
  $\alpha^5$ & 0 & 1 & 1 & 0  \\\hline
  $\alpha^6$ & 1 & 1 & 0 & 0  \\\hline
  $\alpha^7$ & 1 & 0 & 1 & 1  \\\hline
  $\alpha^8$ & 0 & 1 & 0 & 1  \\\hline
  $\alpha^9$ & 1 & 0 & 1 & 0  \\\hline
  $\alpha^{10}$ & 0 & 1 & 1 & 1  \\\hline
  $\alpha^{11}$ & 1 & 1 & 1 & 0  \\\hline
  $\alpha^{12}$ & 1 & 1 & 1 & 1  \\\hline
  $\alpha^{13}$ & 1 & 1 & 0 & 1  \\\hline
  $\alpha^{14}$ & 1 & 0 & 0 & 1  \\
  \hline
\end{tabular}
\end{table}

It is easy to verify that if $S$ is a quintriple, then also $\alpha^i S$ is a quintriple for each $i \geq 0$.
The set of elements $S=\{ \alpha^0,\alpha^1,\alpha^3,\alpha^4,\alpha^7 \}$ is a quintriple,
where $\alpha^4 = \alpha + 1$ and $\alpha^4 = \alpha^3 + \alpha^7$.
The three sets $S$, $\alpha^5 S = \{ \alpha^5,\alpha^6,\alpha^8,\alpha^9,\alpha^{12} \}$,
and $\alpha^{10} S= \{ \alpha^2,\alpha^{10},\alpha^{11},\alpha^{13},\alpha^{14} \}$, form a partition of the nonzero
elements of $F_{2^4}$ and hence also of $\F_2^4$ into three quintriples.
This partition will be called the {\bf \emph{basic partition}}. Partitions have an important role in our exposition.
A basic theorem on partitions will be given later (see Theorem~\ref{thm:spread_parts}), but before the partitions will be based
on lifted rank-metric codes and will be described individually.

Now, let $m=4 r$, $r \geq 2$, and let $\cC$ be a $[4 \times (4r-4),\ell,4]$ MRD code.
By Theorem~\ref{thm:MRD_bound}, we have that $\ell \leq 4r-4$ and there exists a $[4 \times (4r-4),4r-4,4]$ code $\cC$.
Let $\C^{\text{MRD}}$ be the lifted code of~$\cC$.
Each codeword of $\C^{\text{MRD}}$ is a 4-subspace of $\F_2^{4r}$ and any two
such 4-subspaces intersect in the null-space~$\{ {\bf 0}\}$. Hence, each such codeword can be partitioned into three disjoint
quintriples, isomorphic to the quintriples of the basic partition.
The nonzero elements which are contained in these $2^{4r-4}$ 4-subspaces are spanned by the rows of matrices of the form $[I_4~A]$,
where $A$ is a $4 \times (4r-4)$ binary matrix.
Hence, the only nonzero elements of $\F_2^{4r}$ which are not contained in these $2^{4r-4}$
4-subspaces are exactly the $2^{2r-4}-1$ vectors of the form $(0,0,0,0, z_1,z_2, \ldots ,z_{4r-4})$, where $z_i \in \{0,1\}$, and
at least one of the~$z_i$'s is not a \emph{zero}. This set of vectors and the all-zeros vector of length~$4r$,
form a $(4r-4)$-space isomorphic to $\F_2^{4r-4}$. We continue recursively with the same procedure
for $m'=4(r-1)$. The process ends with the initial condition which is
the basic partition of $\F_2^4$. The outcome will be a partition of~$\F_2^{4r}$ into
$\frac{2^{4r}-1}{5}$ pairwise disjoint quintriples. Each one of the $2^{k-2}$ rows of the matrix $T$ is used to construct
one recovery set and each quintriple obtained in this way yields another recovery set. Hence we have the following result
(note that $N_2(2,2)=1$).

\begin{lemma}
\label{lem:sun2}
If $k \geq 2$ and $k \equiv 2~(\text{mod}~4)$ then $N_2(k,2) \geq 2^{k-2} + \frac{2^{k-2}-1}{5}=\frac{3 \cdot 2^{k-1} -1}{5}$.
\end{lemma}

\begin{example}
The following four $4 \times 4$ matrices form a basis for a $[4 \times 4,4,4]$ MRD code $\cC$.
$$
\left[
\begin{array}{cccc}
0 & 0 & 0 & 1 \\
1 & 0 & 0 & 0 \\
1 & 1 & 0 & 1 \\
1 & 0 & 1 & 1
\end{array}
\right],
\left[
\begin{array}{cccc}
1 & 0 & 0 & 0 \\
0 & 1 & 0 & 0 \\
0 & 0 & 1 & 1 \\
0 & 0 & 0 & 1
\end{array}
\right],
\left[
\begin{array}{cccc}
0 & 1 & 0 & 0 \\
0 & 0 & 1 & 0 \\
1 & 0 & 1 & 0 \\
0 & 1 & 0 & 1
\end{array}
\right] ,~ \text{and} ~
\left[
\begin{array}{cccc}
0 & 0 & 1 & 0 \\
0 & 1 & 0 & 1 \\
0 & 0 & 0 & 1 \\
1 & 1 & 0 & 1
\end{array}
\right] .
$$
The code $\cC$ contain sixteen $4 \times 4$ matrices. Each one of these sixteen matrices
is lifted to a 4-subspace of $\F_2^8$ and a code $\C^{\text{MRD}}$ with sixteen 4-subspaces of $\F_2^8$ is obtained.
As an example, the first $4 \times 4$ matrix is lifted to a $4 \times 8$ matrix
$$
\left[
\begin{array}{cccccccc}
1 & 0 & 0 & 0 & 0 & 0 & 0 & 1 \\
0 & 1 & 0 & 0 & 1 & 0 & 0 & 0 \\
0 & 0 & 1 & 0 & 1 & 1 & 0 & 1 \\
0 & 0 & 0 & 1 & 1 & 0 & 1 & 1
\end{array}
\right]
$$
which is a generator matrix of a 4-subspace of $\F_2^8$ whose fifteen nonzero vector are ordered to
have the isomorphism $\varphi$ of this 4-subspace to the additive group of $\F_{16}$, where
$\alpha$ is a root of the primitive polynomial $x^4+x+1$ and $\varphi (\alpha^i)=a_i$
for each $0 \leq i \leq 14$.
$$
\begin{array}{cccccccccc}
a_0 & = & 1 & 0 & 0 & 0 & 0 & 0 & 0 & 1 \\
a_1 & = & 0 & 1 & 0 & 0 & 1 & 0 & 0 & 0 \\
a_2 & = & 0 & 0 & 1 & 0 & 1 & 1 & 0 & 1 \\
a_3 & = & 0 & 0 & 0 & 1 & 1 & 0 & 1 & 1 \\
a_4 = a_0 + a_1 & = & 1 & 1 & 0 & 0 & 1 & 0 & 0 & 1 \\
a_5 = a_1 + a_2 & = & 0 & 1 & 1 & 0 & 0 & 1 & 0 & 1 \\
a_6 = a_2 + a_3 & = & 0 & 0 & 1 & 1 & 0 & 1 & 1 & 0 \\
a_7 = a_0 + a_1 + a_3 & = & 1 & 1 & 0 & 1 & 0 & 0 & 1 & 0 \\
a_8 = a_0 + a_2 & = & 1 & 0 & 1 & 0 & 1 & 1 & 0 & 0 \\
a_9 = a_1 + a_3 & = & 0 & 1 & 0 & 1 & 0 & 0 & 1 & 1 \\
a_{10} = a_0 + a_1 + a_2 & = & 1 & 1 & 1 & 0 & 0 & 1 & 0 & 0 \\
a_{11} = a_1 + a_2 + a_3 & = & 0 & 1 & 1 & 1 & 1 & 1 & 1 & 0 \\
a_{12} = a_0 + a_1 + a_2 + a_3 & = & 1 & 1 & 1 & 1 & 1 & 1 & 1 & 1 \\
a_{13} = a_0 + a_2 + a_3 & = & 1 & 0 & 1 & 1 & 0 & 1 & 1 & 1 \\
a_{14} = a_0 + a_3 & = & 1 & 0 & 0 & 1 & 1 & 0 & 1 & 0 \\
\end{array}
$$
Now, we have a partition of the nonzero vectors of the 4-subspace which is isomorphic to the basic partition as follows:
$$
\{ a_0,a_1,a_3,a_4,a_7 \},~\{ a_5,a_6,a_8,a_9,a_{12} \},~\{ a_2,a_{10},a_{11},a_{13},a_{14} \},
$$
where $a_4=a_0+a_1=a_3+a_7$, $a_9=a_5+a_6=a_8+a_{12}$, and $a_{14}=a_{10}+a_{11}=a_2+a_{13}$.

The same procedure is done on the sixteen 4-subspace of $\C^{\text{MRD}}$ and to the 4-subspace whose generator matrix is
$$
\left[
\begin{array}{cccccccc}
1 & 0 & 0 & 0 & 0 & 0 & 0 & 0 \\
0 & 1 & 0 & 0 & 0 & 0 & 0 & 0 \\
0 & 0 & 1 & 0 & 0 & 0 & 0 & 0 \\
0 & 0 & 0 & 1 & 0 & 0 & 0 & 0
\end{array}
\right] .
$$
These seventeen 4-subspaces of $\F_2^8$ are pairwise disjoint (intersect in the null-space $\{ \bfzero \}$ and hence the ${51 = 17 \cdot 3}$
quintriples which were constructed form a partition of $\F_2^8 \setminus \{ 0 \}$.
$\blacksquare$
\end{example}

We continue and describe the partition of $\F_2^5$. Let $\alpha$ be a root of the primitive
polynomial ${x^5 + x^2+1}$. Each one of the five sets $S_1 = \{\alpha^5,\alpha^0,\alpha^2,\alpha^7,\alpha^{10} \}$,
$S_2 = \alpha S_1$, $S_3 = \alpha^{12} S_1$, $S_4 = \alpha^{13} S_1$,
and ${S_5 = \{ \alpha^{16},\alpha^4,\alpha^{27},\alpha^{29},\alpha^{30} \}}$ is
a quintriple. These five sets and the set $\{ \alpha^9,\alpha^{21},\alpha^{24},\alpha^{25},\alpha^{26},\alpha^{28}\}$
form a partition of the 31 nonzero elements of $\F_2^5$. Moreover, the four elements $\alpha^{21},\alpha^{25},\alpha^{26},\alpha^{28}$
are linearly dependent and will be used later to form an additional recovery set.
Now, let $m=4r+1$ and let $\cC$ be a $[4 \times (4r-3),\ell,4]$ MRD code.
By Theorem~\ref{thm:MRD_bound}, we have that $\ell \leq 4r-3$ and there exists a ${[4 \times (4r-3),4r-3,4]}$ code $\cC$.
Let~$\C^{\text{MRD}}$ be the lifted code of $\cC$.
Each codeword of $\C^{\text{MRD}}$ is a 4-subspace of $\F_2^{4r+1}$ and any two
such 4-subspaces intersect in the null-space $\{ {\bf 0}\}$. Hence, each such codeword can be partitioned into three disjoint
quintriples, isomorphic to the quintriples of the basic partition.
The only nonzero elements of $\F_2^{4r+1}$ which are not contained in these $2^{4r-3}$
4-subspaces are exactly the $2^{4r-3}-1$ vectors of the form $(0,0,0,0, z_1,z_2, \ldots ,z_{4r-3})$, where $z_i \in \{0,1\}$, and
at least one of the $z_i$'s is not \emph{zero}. This set of vectors and the all-zeros vector of length $4r+1$,
form a $(4r-3)$-space isomorphic to $\F_2^{4r-3}$. We continue recursively with the same procedure
for $m'=4r-3$. This recursive procedure ends when $m'=5$, where we use the partition of $\F_2^5 \setminus \{0\}$ with 5 quintriples
and six more elements from which four are linearly dependent. The outcome will be a
partition of $\F_2^{4r+1} \setminus \{0\}$ into $\frac{2^{4r+1}-7}{5}$ quintriples, a subset with four linearly
dependent elements and a subset of size two. The subset with four linearly dependent elements and the remaining element
of the first row yields another recovery set as follows. Let $\{ x_1,x_2,x_3,x_4 \}$ be the set of linearly
dependent elements, i.e., $x_1 + x_2 + x_3 + x_4 =0$ and assume w.l.o.g. that in the first row $(\bfzero,v)$ was not used
in any recovery set. As the remaining
elements in each row can be chosen arbitrarily we choose $(x_1 ,u)$, $(x_2,\bfzero)$, $(x_3,\bfzero)$, and $(x_4,\bfzero)$.
Since $(\bfzero,u) = (x_1 , u) + (x_2,\bfzero) + (x_3,\bfzero) + (x_4,\bfzero)$ and $(\bfzero,v)$ from the first row was not used
in a recovery set, it follows that $U$ is recovered from $\{ (\bfzero,v), (x_1,u) ,(x_2,\bfzero),(x_3,\bfzero),(x_4,\bfzero) \}$.
Thus, the construction leads to the following bound (note that $N_2(3,2)=2$).

\begin{lemma}
\label{lem:sun3}
If $k \geq 3$ and $k \equiv 3~(\text{mod}~4)$ then $N_2(k,2) \geq 2^{k-2} + \frac{2^{k-2}-7}{5} +1=\frac{3 \cdot 2^{k-1} -2}{5}$.
\end{lemma}

We continue with an example of the partition for $\F_2^6$.

\begin{example}
\label{ex:quintriples6}
Let $\alpha$ be a root of the primitive polynomial $x^6 + x+1$, i.e., $\alpha^6 = \alpha +1$. The table of $\F_{63}$ yields
the following equalities
\begin{align*}
\alpha^0 = \alpha^1 + \alpha^6 = \alpha^{13} + \alpha^{35}, \\
\alpha^9 = \alpha^7 + \alpha^{19} = \alpha^{12} + \alpha^{41} .
\end{align*}
These two equalities yield the following four quintriples
\begin{align*}
S_1 &= \{\alpha^0,\alpha^1,\alpha^6,\alpha^{13},\alpha^{35} \},\\
S_2 = \alpha^2 S_1 &= \{\alpha^2,\alpha^3,\alpha^8,\alpha^{15},\alpha^{37} \},\\
S_3 = \alpha^4 S_1 &= \{\alpha^4,\alpha^5,\alpha^{10},\alpha^{17},\alpha^{39} \},\\
S_4 &= \{\alpha^7,\alpha^9,\alpha^{12},\alpha^{19},\alpha^{41} \}.\\
\end{align*}
It is east to verify that for each $j$, $0 \leq j \leq 20$, $j \neq 11$, there is exactly one $i$,
such that ${i \equiv j~ (\mmod~21)}$ and $\alpha^i \in S_1 \cup S_2 \cup S_3 \cup S_4$; for $j =11$, there
is no ${i \equiv j~ (\mmod~21)}$ such that
$\alpha^i \in S_1 \cup S_2 \cup S_3 \cup S_4$. This implies that
the twelve sets $S_1,S_2,S_3,S_4, \alpha^{21} S_1, \alpha^{21} S_2, \alpha^{21} S_3, \alpha^{21} S_4,
\alpha^{42} S_1, \alpha^{42} S_2, \alpha^{42} S_3, \alpha^{42} S_4$, and the 2-subspace $\{{\bf 0},\alpha^{11},\alpha^{32},\alpha^{53}\}$
form a partition of $\F_2^6$ into 12 quintriples and one 2-subspace.
$\blacksquare$
\end{example}

The leftover in each row of the 2-subspace, $\{{\bf 0},\alpha^{11},\alpha^{32},\alpha^{53}\}$, of Example~\ref{ex:quintriples6}, and the leftover
element from the first row, $(\bfzero,v)$, can be used to form a recovery set which consists of the elements
$\{(\bfzero,v),(\alpha^{11},\bfzero),(\alpha^{32},\bfzero),(\alpha^{53},u)\}$.
Hence, starting with $m=4r+2$ the partition $\F_2^{4r+2} \setminus \{0\}$ yields the following bound (note that $N_2(4,2)=5$).

\begin{lemma}
\label{lem:sun0}
If $k \geq 4$ and $k \equiv 0~(\text{mod}~4)$ then $N_2(k,2) \geq 2^{k-2} + \frac{2^{k-2}-4}{5} +1=\frac{3 \cdot 2^{k-1} +1}{5}$.
\end{lemma}

The construction for $m=4r+3$ is similar, but the construction for the initial condition of $m=7$
is more complicated. We manage to obtain a
partition of $\F_2^{4r+3} \setminus \{0\}$ into $\frac{2^{4r+3}-8}{5}$ quintriples, a subset with four linearly
dependent elements and another subset of size three. The subset with four linearly dependent elements and the remaining element
of the first row yields another recovery set and hence we obtain the following bound (note that $N_2(5,2)=9$).

\begin{lemma}
\label{lem:sun1}
If $k \geq 5$ and $k \equiv 1~(\text{mod}~4)$ then $N_2(k,2) \geq 2^{k-2} + \frac{2^{k-2}-8}{5} +1=\frac{3 \cdot 2^{k-1} -3}{5}$.
\end{lemma}

Lemmas~\ref{lem:sun2},~\ref{lem:sun3},~\ref{lem:sun0}, and~\ref{lem:sun1} imply the following lower bound.

\begin{corollary}
\label{cor:two_lower}
If $k \geq 2$ is a positive integer, then $N_2 (k,2) \geq \left\lfloor \frac{3 \cdot 2^{k-1} +1}{5} \right\rfloor$.
\end{corollary}

Corollary~\ref{cor:two_lower} and Lemma~\ref{lem:two_upper} imply the following theorem.

\begin{theorem}
\label{thm:two_dim2}
For $k \geq 2$, $N_2 (k,2) = \left\lfloor \frac{3 \cdot 2^{k-1} +1}{5} \right\rfloor$.
\end{theorem}

\subsection{Recovery Sets for 4-Subspaces}
\label{sec:large_sets}

Assume that the 4-subspace to be recovered is $U$ which is spanned by the vectors $(\bfzero,u_1)$, $(\bfzero,u_2)$, $(\bfzero,u_3)$,
and $(\bfzero,u_4)$ ($U$ will be referred as $\F_2^4$).

When $d=4$, the first row of $T$ can be partitioned into three recovery sets of size 4 and three leftovers.
Each other row of $T$ is partitioned into three recovery sets of size 5 and one leftover.

Consider now the 3-subspace of $\F_2^{k-4}$ spanned by $\{ x_1,x_2,x_3 \}$, where $x_1,x_2,x_3 \in \F_2^{k-4}$.
Consider now the seven vectors of $\F_2^k$,
$$
R \triangleq \{(x_1,u_1),(x_2,u_1+u_2),(x_3,u_1+u_3),(x_1+x_2,\bfzero),
$$
$$
(x_1+x_3,\bfzero),(x_2+x_3,u_2+u_3+u_4),(x_1+x_2+x_3,u_2+u_3)\} .
$$
These seven vectors form a recovery set for $U$ since we can recover
$(\bfzero,u_1)$, $(\bfzero,u_2)$, $(\bfzero,u_3)$, and $(\bfzero,u_4)$ as follows:
$$
(\bfzero,u_1)=(x_1,u_1)+(x_2,u_1+u_2)+(x_3,u_1+u_3)+(x_1+x_2+x_3,u_2+u_3)
$$
$$
(\bfzero,u_2)=(x_1,u_1)+(x_2,u_1+u_2)+(x_1+x_2,\bfzero)
$$
$$
(\bfzero,u_3)=(x_1,u_1)+(x_3,u_1+u_3)+(x_1+x_3,\bfzero)
$$
$$
(\bfzero,u_4)=(x_2,u_1+u_2)+(x_3,u_1+u_3)+(x_2+x_3,u_2+u_3+u_4).
$$
These seven elements of $R$ are formed from the rows indexed by
$$
x_1, ~ x_2, ~ x_3 , ~ x_1 + x_2, ~ x_1 + x_3, ~ x_2 + x_3, ~ x_1 + x_2 + x_3,
$$
in $T$, which are associated with exactly one 3-subspace of $\F_2^{k-4}$. This type of recovery set will be called
a {\bf \emph{(3,4)-recovery set model}}.

Since each internal row of $T$ yields exactly one leftover and this leftover can be chosen to be any element in the row,
it follows that to form more recovery sets we have to partition $\F_2^{k-4}$ to many 3-subspaces and some additional
elements which can be used with the three leftovers of the first row to form additional recovery sets. We distinguish between
three cases, depending on whether $k$ is congruent to  1, 2, or 0 modulo 3.

\noindent
{\bf Case 1:} $k \equiv 1 ~(\mmod ~ 3)$, i.e., $k-4 \equiv 0 ~(\mmod ~ 3)$. (note that 3 divides $k-7$.)

By Theorem~\ref{thm:MRD_bound}, there exists a $[3 \times (k-7),k-7,3]$ MRD code $\cC$.
Let $\C^{\text{MRD}}$ be the lifted code of~$\cC$.
Each codeword of $\C^{\text{MRD}}$ is a 3-subspace of $\F_2^{k-4}$ and any two
such 3-subspaces intersect in the null-space $\{ {\bf 0}\}$. By using the (3,4)-recovery set model, each such codeword can be used to form
a recovery set from the leftovers of the associated rows of $T$.
The nonzero elements which are contained in these $2^{k-7}$ 3-subspaces are spanned by the rows of matrices of the form $[I_3~A]$,
where $A$ is a $3 \times (k-7)$ binary matrix.
Therefore, the only nonzero elements of $\F_2^{k-4}$ which are not contained in these $2^{k-4}$
3-subspaces are exactly the $2^{k-7}-1$ vectors of the form $(0,0,0, z_1,z_2, \ldots ,z_{k-7})$, where $z_i \in \{0,1\}$, and
at least one of the~$z_i$'s is not a \emph{zero}. This set of vectors and the all-zeros vector of length~$k-4$,
form a $(k-4)$-space isomorphic to $\F_2^{k-4}$. We continue recursively with the same procedure.
The process ends with the initial condition which is the 3-subspace $\F_2^3$ from which one recovery set is obtained.
The outcome will be a partition of $\F_2^{k-4}$ into
$\frac{2^{k-4}-1}{7}$ pairwise disjoint 3-subspaces. Each one of the $2^{k-4}$ rows of the matrix $T$
is used to construct three recovery sets. Hence we have the following result.

\begin{lemma}
\label{lem:lb_d=4c1}
If $k >6$ and $k \equiv 1~(\textup{mod}~3)$ then $N_2(k,4) \geq 3 \cdot 2^{k-4} + \frac{2^{k-4}-1}{7}=\frac{11 \cdot 2^{k-3} -1}{7}$.
\end{lemma}

\noindent
{\bf Case 2:} $k \equiv 2 ~(\mmod ~ 3)$, i.e., $k-4 \equiv 1 ~(\mmod ~ 3)$.

Using the same technique as in Case 1, the process ends with the initial condition which is the 4-subspace $\F_2^4$.
From this 4-subspace, one 3-subspace can be obtained to form one recovery set. From the eight remaining elements,
4 other linearly dependent elements can be used to obtain any element of $U$, which together with the three leftovers from the first row of $T$
form another recovery set for the 4-subspace~$U$. Thus, we have the following lemma.

\begin{lemma}
\label{lem:lb_d=4c2}
If $k > 6$ and $k \equiv 2~(\textup{mod}~3)$, then $N_2(k,4) \geq 3 \cdot 2^{k-4} + \frac{2^{k-4}-16}{7} +2 =\frac{11 \cdot 2^{k-3} -2}{7}$.
\end{lemma}

\noindent
{\bf Case 3:} $k \equiv 0 ~(\mmod ~ 3)$, i.e., $k-4 \equiv 2 ~(\mmod ~ 3)$.

Using the same technique as in Case 1 and Case 2, the process ends with the initial condition which is the 5-subspace $\F_2^5$.
From this 5-subspace one 3-subspace can be obtained to form one recovery set and the other 24 leftovers with the three leftovers
from the first row of $T$ can be used to obtain three more recovery sets for the 4-subspace $U$. The process can also be ended
with the initial condition of the 8-subspace~$\F_2^8$ for which there are 34 disjoint 3-subspaces and 17 remaining elements.
With the three leftovers from the first row of $T$, we can construct two recovery sets to have the same result.
Thus, we have the following lemma.

\begin{lemma}
\label{lem:lb_d=4c0}
If $k > 6$ and $k \equiv 0~(\textup{mod}~3)$, then $N_2(k,4) \geq 3 \cdot 2^{k-4} + \frac{2^{k-4}-32}{7} +4 =\frac{11 \cdot 2^{k-3} -4}{7}$.
\end{lemma}

Lemmas~\ref{lem:lb_d=4c1}, \ref{lem:lb_d=4c2}, and \ref{lem:lb_d=4c0}, imply the following lower bound.

\begin{theorem}
\label{thm:low_bound_4}
For $k \geq 7$,
$$
N_2(k,4) \geq \left\lfloor \frac{11 \cdot 2^{k-3} -1}{7}  \right\rfloor .
$$
\end{theorem}

The case of $k=6$ is solven in the following example.
\begin{example}
\label{ex:d=4_k=6}
For $k=6$ the matrix $T$ has four rows indexed by $0,\beta^0,\beta,\beta^2$, where $\beta$ in a primitive element is $\F_4$.
There are 16 columns in $T$ indexed by $0,\alpha^0,\alpha,\alpha^2,\alpha^3,\ldots,\alpha^{14}$.
The first three recovery sets are taken from the first row of $T$ as follows.
$$
\{ (\bfzero,\alpha^3),(\bfzero,\alpha^4),(\bfzero,\alpha^5),(\bfzero,\alpha^6)\},
\{ (\bfzero,\alpha^7),(\bfzero,\alpha^8),(\bfzero,\alpha^9),(\bfzero,\alpha^{10})\},
$$
$$
\{ (\bfzero,\alpha^{14}),(\bfzero,\alpha^0),(\bfzero,\alpha^1),(\bfzero,\alpha^2)\}.
$$
The next nine recovery set are taken from the other rows of $T$, three from each row as follows.
$$
\{ (\beta^i,\alpha^4),(\beta^i,\alpha^5),(\beta^i,\alpha^6),(\beta^i,\alpha^7),(\beta^i,\alpha^8)\},~~~
\{ (\beta^i,\alpha^9),(\beta^i,\alpha^{10}),(\beta^i,\alpha^{11}),(\beta^i,\alpha^{12}),(\beta^i,\alpha^{13})\},
$$
$$
\{ (\beta^i,\bfzero),(\beta^i,\alpha^0),(\beta^i,\alpha^1),(\beta^i,\alpha^2),(\beta^i,\alpha^3)\},
$$
where $0 \leq i \leq 2$.

To these twelve recovery set we add the set
$$
\{ (\bfzero,\alpha^{11}),(\bfzero,\alpha^{12}),(\bfzero,\alpha^{13}),(\beta^0,\alpha^{14}),(\beta^1,\alpha^{14}),(\beta^2,\alpha^{14})\}.
$$
Since $(\beta^0,\alpha^{14})+(\beta^1,\alpha^{14})+(\beta^2,\alpha^{14})=( \bfzero ,\alpha^{14})$ the last set is also a recovery set.

It can be easily verified that we cannot form more than 13 recovery sets.
$\blacksquare$
\end{example}

\begin{corollary}
\label{cor:d=4_k=6}
$N_2(6,4) = 13$.
\end{corollary}

The last two cases are trivial and are given in the following lemma.
\begin{lemma}
$N_2(5,4) = 6$ and $N_2(4,4) = 3$.
\end{lemma}

\subsection{Recovery Sets for 5-Subspaces}
\label{sec:large_sets5}

Assume that the 5-subspace to be recovered is $U$ which is spanned by the vectors $(\bfzero,u_1)$, $(\bfzero,u_2)$, $(\bfzero,u_3)$,
$(\bfzero,u_4)$ and $(\bfzero,u_5)$ ($U$ will be referred as $\F_2^5$).

When $d=5$, the first row of $T$ can be partitioned into six recovery sets of size 5 and one leftover.
Each other row of $T$ is partitioned into five recovery sets of size 6 and two leftovers.

Consider now four disjoint 2-subspaces of $\F_2^{k-5}$ spanned by $\{ x_1,x_2 \}$, $\{ x_3,x_4 \}$, $\{ x_5,x_6 \}$, and $\{ x_7,x_8 \}$.
The 24 leftovers of these four 2-subspaces are chosen and partitioned into the following three subsets of size eight:
\begin{align*}
& \{(x_1,\bfzero),(x_2,\bfzero),(x_1+x_2,u_3),(x_7,\bfzero),(x_1,u_1),(x_2,u_2),(x_1+x_2,u_4),(x_7,u_5)\}\\
& \{(x_3,\bfzero),(x_4,\bfzero),(x_3+x_4,u_3),(x_8,\bfzero),(x_3,u_1),(x_4,u_2),(x_3+x_4,u_4),(x_8,u_5)\}\\
& \{(x_5,\bfzero),(x_6,\bfzero),(x_5+x_6,u_3),(x_7+x_8,\bfzero), (x_5,u_1),(x_6,u_2),(x_5+x_6,u_4),(x_7+x_8,u_5)\}.
\end{align*}
It can be readily verified that each one of these three subsets is a recovery set for $U$.

Each internal row of $T$ yields two leftovers and if one of them is chosen as $(x,\bfzero)$ then the second one can be chosen arbitrarily.
But, if $(x,u_3)$ was chosen, then $(x,u_4)=(x, \alpha u_3)$ must be the next element in the row, where $\alpha$ is
a primitive element of $\F_{2^d}$.
Therefore, to obtain these three recovery sets from the leftovers we have to partition $\F_2^{k-5}$ into 2-subspaces and to partition
these 2-subspaces into subsets with four 2-subspaces in each one. We distinguish between odd and even $k$.

\noindent
{\bf Case 1:} If $k$ is odd, then $k-5$ is even and hence $\F_2^{k-5}$ can be partitioned into $\frac{2^{k-5}-1}{3}$ ~ 2-subspaces.
These disjoint 2-subspaces can be partitioned into $\frac{2^{k-5}-4}{12}=\frac{2^{k-7}-1}{3}$ sets of four 2-subspaces and another
2-subspace. Each such set of four 2-subspaces forms three recovery sets for $U$ for a total of $2^{k-7}-1$ recovery sets.
Assume that the additional 2-subspace is $\{ x_1,x_2,x_1+x_2 \}$ and the leftover from the first row of $T$ is $(\bfzero,u_1)$.
The following set of seven 1-subspaces
$$
\{ (\bfzero,u_1),(x_1,\bfzero),(x_1,u_2),(x_2,\bfzero),(x_2,u_3),(x_1+x_2,u_4),(x_1+x_2,u_5) \}
$$
is another recovery set. Thus, we have that
\begin{lemma}
\label{lem:low_bound5a}
If $k\geq 7$ and $k$ is odd, then $N_2(k,5)\ge 6+5(2^{k-5}-1)+2^{k-7}-1 +1=21\cdot2^{k-7} +1$.
\end{lemma}

\noindent
{\bf Case 2:} If $k$ is even, then $k-5$ is odd and hence $\F_2^{k-5}$ can be partitioned into $\frac{2^{k-5}-8}{3}$ ~ 2-subspaces
and one 3-subspace.
These disjoint 2-subspaces can be partitioned into $\frac{2^{k-5}-8}{12}=\frac{2^{k-7}-2}{3}$ sets of four 2-subspaces.
Each such set of four 2-subspaces
forms three recovery sets for $U$ for a total of $2^{k-7}-2$ recovery sets. Assume that the remaining 3-subspace is $\Span{x_1,x_2,x_3}$
and one leftover from the first row of $T$ is $(\bfzero,u_1)$. They can form the following two recovering sets:
$$
\{ (\bfzero,u_1),(x_1+x_2,\bfzero),(x_1+x_2,u_2),(x_1+x_3,\bfzero),(x_1+x_3,u_3),(x_2+x_3,u_4),(x_2+x_3,u_5) \},
$$
$$
\{ (x_1,\bfzero),(x_1,u_1),(x_2,\bfzero),(x_2,u_2),(x_3,\bfzero),(x_3,u_3),(x_1+x_2+x_3,u_4),(x_1+x_2+x_3,u_5) \}.
$$
Thus, we have that
\begin{lemma}
\label{lem:low_bound5b}
If $k\geq 7$ and $k$ is even, then $N_2(k,5)\geq 6+5(2^{k-5}-1)+2^{k-7}-2 +2 =21\cdot2^{k-7} +1$.
\end{lemma}

\section{Upper Bounds on the Size of Recovery Sets for Binary Alphabet}
\label{sec:upper_binary}

We turn now to consider an upper bound on $N_2(k,4)$. We will have to distinguish again between three cases depending on whether $k$
is congruent to 0, 1, or 2 modulo 3. We start by ignoring the first row of~$T$ and
consider the other rows of $T$. From each such row, we have constructed 3 recovery sets of size 5 and we had one leftover.
The $2^{k-4}-1$ leftovers we used to form recovery sets of size 7 based on disjoint 3-subspaces. The number of disjoint
3-subspaces depends on the value of $k-4$ modulo 3. If $k-4 \equiv 0 ~(\mmod ~ 3)$, then there are $\frac{2^{k-4}-1}{7}$ such
3-subspaces. If $k-4 \equiv 1 ~(\mmod ~ 3)$, then there are $\frac{2^{k-4}-9}{7}$ such
3-subspaces and eight 1-subspaces are left out of this partition. If $k-4 \equiv 2 ~(\mmod ~ 3)$, then there are $\frac{2^{k-4}-25}{7}$ such
3-subspaces and twenty four 1-subspaces are left out of this partition.

We start to check if we can construct other recovery sets by considering combinations of elements from different rows of $T$.
We cannot gain anything by considering seven elements from different rows as this is already guaranteed by the
construction of recovery sets of size seven from the leftovers. Hence, we will concentrate on recovery sets of size six
to replace some recovery sets of size seven. But, each recovery set of size six from a few rows will require at least
two elements from two rows. This will imply that we will have to give up on a recovery set of size five from each one of these two rows.
Hence, we further need to use these rows to form more recovery sets of size six. Therefore,
we are going to waste at least two elements in these rows which will make it worse than using the only one leftover in each
row for a recovery set of size seven. Therefore, if we want to gain something beyond our construction it is required
to use the three leftovers from the first row carefully.

Now consider
to use the three leftovers from the first row (or even more elements from the first row in a recovery set that contains elements
from other rows). We distinguish between three cases, depending on whether an associated recovery set contains one element from
the first row, two elements from the first row, or three elements from the first row.

\noindent
{\bf Case 1:} Only one element from the first row is contained in the recovery set.
Assume first that all the other elements of the recovery set are from different rows of $T$.
But, this cannot be done with less than six leftovers from other rows and hence the recovery set will contain
at least seven elements which does not save any element to obtain more recovery sets.
If from one row we use more than two elements, then this row cannot be used anymore to form three recovery sets without
elements from other rows. This implies that for each element from the first row used in a recovery set used with four
elements from another row, we replace one leftover from the first row with one leftover from another row, with whom it is
less flexible to form recovery sets. If we use elements from two rows then at least six elements are used in the recovery set
and at least two other rows will be used so we will lose the two leftovers from the first row without any gain in the best case.

\noindent
{\bf Case 2:} Two elements from the first row are in one recovery set. If we use only the unique leftover from each other row, then
the recovery set will be with at least seven elements which does not save any elements as in Case 1. If we use three elements
from one row, then we can again use arguments similar to Case 1 to exclude this case. The same goes for generating a recovery set
with six elements, two from each other row.

\noindent
{\bf Case 3:} All the three elements from the first row are in one recovery set. By the previous analysis, it is clear that
at best we can use three leftovers. But, this does not leave us a better option for the other recovery sets
than the ones used in the constructions.

This analysis implies that the lower bound of Theorem~\ref{thm:low_bound_4}
is tight and hence we proved the following theorem.

\begin{theorem}
\item For $k \geq 7$,
$$
N_2(k,4) = \left\lfloor \frac{11 \cdot 2^{k-3} -1}{7}  \right\rfloor .
$$
\end{theorem}

The analysis done in this section can be implemented for recovery sets with higher dimensions with extra complexity.
By applying the integer programming technique with the careful analysis as was done for $d=2$ and $d=4$ we obtain the
following upper bound,
$$
N_2(k,5) \leq 21 \cdot 2^{k-7} +2 .
$$
for $k \geq 7$. Combining Lemmas~\ref{lem:low_bound5a} and~\ref{lem:low_bound5b}, this completes the proof of the following theorem.

\begin{theorem}
\label{thm:bound_k_5}
\item For $k \geq 7$,
$$
21 \cdot 2^{k-7}+1 \leq N_2(k,5) \leq 21 \cdot 2^{k-7} +2.
$$
\end{theorem}

\section{Bounds for a Larger Alphabet}
\label{sec:general}

When $q>2$, the defined matrix $T$ for $q=2$ is no longer appropriate to represent
the elements of PG$(k-1,q)$. We modify the representation by defining a $\frac{q^{k-d}-1}{q-1} \times q^d$ matrix $T$ whose rows
are indexed by the points of PG$(k-d-1,q)$ (or the 1-subspaces of $V_q^{k-d}$). It is the same as identifying these rows
by words of length $k-d$ whose first nonzero entry is a \emph{one}.
The columns of the matrix $T$ are indexed by the elements of $\F_{q^d}$, where the first column is indexed by $\bfzero$ and the others
by $\alpha^0,\alpha^1,\ldots,\alpha^{q^d-2}$, in this order, where $\alpha$ is a primitive element in $\F_{q^d}$.
This matrix is the same as the one defined for $q=2$ with the omission of the first row of $T$ defined for $q=2$.
Entry $T(x,y)$ in the matrix $T$ can be represented by $(x,y)$ or $x+y$.
To this matrix $T$ we add a vector $T_d$ whose length is $\frac{q^d-1}{q-1}$ and whose elements represent the $\frac{q^d-1}{q-1}$ one-dimensional
subspaces of $V_q^d$. These elements can be taken as the ones in the set $\{ \alpha^i ~:~ 0 \leq i < \frac{q^d-1}{q-1} \}$.
It is readily verified that the entries of $T$ and~$T_d$ together represent all the one-dimensional subspaces of $V_q^k$.

\subsection{Constructions for Recovery Sets}

We turn our attention now to construct recovery sets for $d$-subspaces, $d \geq 2$, when $q>2$.
In this subsection, we first provide a general construction with recovery sets of size $d$ and $d+1$ as before
and the leftovers from $T$ will be combined to form recovery sets of size $d+2$.

We start by considering the case of $d=2$.
Note first that $\frac{q^2 -1}{q-1} = q+1$, i.e., a 2-subspace contains $q+1$ 1-subspaces.
Each $2$ consecutive elements in $T_2$ span $\F_q^2$ and hence we obtain $\left\lfloor \frac{q+1}{2} \right\rfloor$ recovery sets
from $T_2$. From each other row, we will have $\left\lfloor \frac{q^2}{3} \right\rfloor$ recovery sets and one or two leftovers, unless
$q$ is divisible by 3. This implies that we have a tight bound when $q \equiv 0~(\mmod ~ 3)$ as implied
by Theorem~\ref{thm:tight} and when $q \not\equiv 0~(\mmod ~ 3)$ we have
to apply a construction for the leftovers similar to the one in Section~\ref{sec:small_sets}.

When $d$ divides $\frac{q^d-1}{q-1}$ there are no leftovers in $T_d$ and if $d$ does not divides $\frac{q^d-1}{q-1}$, then
the number of leftovers in $T_d$ is the reminder from the division of $\frac{q^d -1}{q-1}$ by $d$, i.e.,
an integer between 1 and $d-1$.
The matrix~$T$ is of size $\frac{q^{k-d}-1}{q-1} \times q^d$ and by Lemma~\ref{lem:simple_recover}
each $d+1$ consecutive elements in a row span $\F_q^d$.
If there are no leftovers in a row of $T$, then we have the following theorem which is a special case of Theorem~\ref{thm:tight}.
\begin{theorem}
If $d+1$ divides $q^d$, then $N_q (k,d) = \left\lfloor \frac{q^d -1}{d(q-1)} \right\rfloor + \frac{q^k-q^d}{(d+1)(q-1)}$.
\end{theorem}

For the next construction and its analysis, it will be required to use the well-known partitions of vector spaces
and projective geometries~\cite{Seg64}.
Such partitions can be found also in~\cite{EtVa11,ScEt02} and other places as well. For completeness, one such partition is presented
in the following theorem.

\begin{theorem}
\label{thm:spread_parts}
If $d$ divides $n$, then $\F_q^n$ can be partitioned into $\frac{q^n-1}{q^d-1}$ pairwise disjoint (intersect in the null-space $\{ {\bf 0}\}$)
$d$-subspaces.
\end{theorem}
\begin{proof}
Let $\alpha$ be a primitive element in $\F_{q^n}$ and let $r = \frac{q^n-1}{q^d-1}$. For each $i$, $0 \leq i \leq r-1$, define
$$
S_i \triangleq \{ \alpha^i , \alpha^{r+i},\alpha^{2r+i},\ldots, \alpha^{(q^d-2)r+i} \}~.
$$
It is easy to verify that $S_0 \cup \{ 0 \} = \F_{q^d}$ and hence $S_0$ is closed under addition. Therefore,
each $S_i \cup \{ 0 \}$ is closed under addition and hence it forms a subspace. Finally, the claim of the theorem
is obtained by the isomorphism between $\F_{q^n}$ and $\F_q^n$.
\end{proof}
Another proof of Theorem~\ref{thm:spread_parts} is by using lifted MRD codes exactly as in the constructions
for the recovery sets, where $d=2$ or $d=4$. More information on these partitions for other parameters can  be found in~\cite{EtSt16}.

If there are leftovers in a row of $T$, then the number of leftovers in such a row is
$q^d - \left\lfloor \frac{q^d}{d+1} \right\rfloor (d+1)$ and in all the $\frac{q^{k-d} -1}{q-1}$ rows of $T$ the
total number of leftovers is
$$
\frac{q^{k-d} -1}{q-1} \left( q^d - \left\lfloor \frac{q^d}{d+1} \right\rfloor (d+1) \right)~.
$$
We choose all the leftovers in any row of $T$ to be consecutive elements in the row as will be described now.
Assume for simplicity that $d+2$ divides $q+1$ and also that $k-d$ is even. Since $k-d$ is even, it follows that the elements of
PG$(k-d-1,q)$ can be partitioned into 2-subspaces of $\F_q^{k-d}$ (1-subspaces in the projective geometry, i.e.,
lines of PG$(k-d-1,q)$) (see Theorem~\ref{thm:spread_parts}).
Each such 1-subspace has $q+1$ points. Since $d+2$ divides $q+1$, it follows that each such 1-subspace can be partitioned
into $\frac{q+1}{d+2}$ subsets of size ${d+2}$. Let $\{ x_1, x_2, \ldots ,x_{d+2} \}$ be such a subset.
Consider the following $d+2$ leftovers in these $d+2$ rows, one leftover for a row as follows, $(x_1,u_1)$, $(x_2,u_2)$,...,
$(x_d,u_d)$, $(x_{d+1},\bfzero)$, $(x_{d+2},\bfzero)$, where $u_1,u_2,\ldots,u_d$ are $d$ linearly independent elements of $\F_q^d$.
For any $i$, $1 \leq i \leq d$, we have that $x_i$, $x_{d+1}$, $x_{d+2}$ are linearly dependent since they are all contained
in the same 1-subspace. Hence, from $(x_i,u_i)$, $(x_{d+1},\bfzero)$, $(x_{d+2},\bfzero)$ we can recover $(\bfzero,u_i)$.
Therefore, $(\bfzero,u_1)$, $(\bfzero,u_2)$,...,$(\bfzero,u_d)$ can be recovered and hence $U$ can be recovered.
If each row of $T$ has exactly one leftover, then this is enough and we have proved the following lemma.

\begin{lemma}
\label{lem:SPoptimal+d+2-q+1}
If $q^d \equiv 1~(\text{mod}~d+1)$, $d+2$ divides $q+1$, and $k-d$ is even, then
$$
N_q (k,d) \geq \left\lfloor \frac{q^d -1}{d(q-1)} \right\rfloor + \frac{q^{k-d}-1}{q-1} \left\lfloor \frac{q^d}{d+1} \right\rfloor
+ \frac{q^{k-d}-1}{(d+2)(q-1)}
$$
\end{lemma}

Based on the analysis done before, we know that we cannot have more recovery sets of size $d$ or $d+1$ and hence
it is easily verified that the lower bound of Lemma~\ref{lem:SPoptimal+d+2-q+1} is tight, i.e.,
\begin{theorem}
\label{thm:SPoptimal+d+2-q+1}
If $q^d \equiv 1~(\text{mod}~d+1)$, $d+2$ divides $q+1$, and $k-d$ is even, then
$$
N_q (k,d) = \left\lfloor \frac{q^d -1}{d(q-1)} \right\rfloor + \frac{q^{k-d}-1}{q-1} \left\lfloor \frac{q^d}{d+1} \right\rfloor
+ \frac{q^{k-d}-1}{(d+2)(q-1)}
$$
\end{theorem}

Theorem~\ref{thm:SPoptimal+d+2-q+1} is a special case of the next theorem which will be analyzed now.

Assume now that each row of $T$ has at least two leftovers. Again, since $d+2$ divides $q+1$, it follows that each 1-subspace, in the
partition into 2-subspaces of $\F_q^{k-d}$, can be partitioned
into $\frac{q+1}{d+2}$ subsets of size $d+2$. Let $\{ x_1, x_2, \ldots ,x_{d+2} \}$ be such a subset.
Consider the following $d+2$ leftovers in these $d+2$ rows, one leftover for a row as follows, $(x_1,u_1)$, $(x_2,u_2)$,...,
$(x_d,u_d)$, $(x_{d+1},u_1)$, $(x_{d+2},u_1)$, where $u_1,u_2,\ldots,u_d$ are $d$ linearly independent elements.
We can recover the two elements $(x_{d+1} - x_1,\bfzero)$ and $(x_{d+2} - x_1,\bfzero)$.
For any $i$, $1 \leq i \leq d$, we have that $x_i$, $x_{d+1} -x_1$, and $x_{d+2} -x_1$ are linearly dependent since they are all contained
in the same 1-subspace. Hence, from $(x_i,u_i)$, $(x_{d+1}-x_1,\bfzero)$, $(x_{d+2}-x_1,\bfzero)$ we can recover $(\bfzero,u_i)$.
Therefore, $(\bfzero,u_1)$, $(\bfzero,u_2)$,...,$(\bfzero,u_d)$ can be recovered and hence $U$ can be recovered too. Now, we choose a second
leftover for each of these rows, $(x_1,\alpha u_1)$, $(x_2,\alpha u_2), \cdots,(x_d,\alpha u_d)$, $(x_{d+1},\alpha u_1)$, $(x_{d+2},\alpha u_1)$.
Since the columns are indexed by consecutive powers of $\alpha$, it follows that each two elements chosen in these rows are
consecutive and hence they can be chosen as leftovers. The subset $U$ is recovered from these $d+2$ leftovers in the
same way as it was recovered from the first $d+2$ leftovers that were chosen. If there are $\ell$ leftovers in a row, then in these rows
we choose the leftovers as
$(x_1,\alpha^i u_1)$, $(x_2,\alpha^i u_2)$,...,$(x_d,\alpha^i u_d)$, $(x_{d+1},\alpha^i u_1)$, $(x_{d+2},\alpha^i u_1)$, for
each $i$, $0 \leq i \leq \ell -1$. For each such $i$, the subspace $U$ is recovered in the same way.

This construction implies the proof for the following theorem.

\begin{theorem}
\label{thm:optimal+d+2-q+1}
If $q \in \dP$, $q > 2$, $d>1$, $k-d$ even, and $q+1 = \ell (d+2)$, then
$$
N_q (k,d) = \left\lfloor \frac{q^d -1}{d(q-1)} \right\rfloor + \frac{q^{k-d}-1}{q-1} \left\lfloor \frac{q^d}{d+1} \right\rfloor
+ \frac{q^{k-d}-1}{q-1} \left( \frac{q^d}{d+2} - \left\lfloor \frac{q^d}{d+1} \right\rfloor \frac{d+1}{d+2}   \right) .
$$
\end{theorem}

A similar analysis can be done when
$k-d$ is odd. The main difference from even $k-d$ is that when $k-d$ is odd the elements of
PG$(k-d-1,q)$ can be partitioned into 2-subspaces of $\F_q^{k-d}$ and $q^2$ 1-subspaces~\cite{Beu75,HoPa72}.

\subsection{Upper Bound on the Number of Recovery Sets}
\label{sec:upper_bounds}

For more upper bounds on the number of recovery sets, we will not distinguish between the binary and the non-binary case.
We have already demonstrated upper bounds based on integer programming and direct analysis. We will develop now another
upper bound based on a direct analysis of possible sizes of recovery sets.
We must use at least $d$ elements from $T_d$ to recover a given $d$-subspace.
From any given row of $T$ at least $d+1$ elements are required to recover a subspace.
Any subset of $d$ elements from two rows of $T$, with at least one element from each row span a subspace~$\cV$, where
$\dim (\cV \cap \F_q^d) \leq d-2$ and hence at least $d+2$ elements are required if the recovery set contains
at least one element from two distinct rows of $T$.
This will immediately improve the bound of Theorem~\ref{thm:gen_upper} to

\begin{theorem}
\label{thm:first_upper_bound}
If $q \in \dP$ and $k \geq d$ is a positive integer, then
$$
N_q (k,d) \leq \left\lfloor \frac{q^d-1}{d(q-1)} \right\rfloor + \frac{q^{k-d}-1}{q-1} \left\lfloor \frac{q^k}{d+1} \right\rfloor +
\left\lfloor \frac{2\ell +\frac{q^{k-d}-1}{q-1}t}{d+2} \right\rfloor ~.
$$
where $\ell$ is the reminder from the division of $\frac{q^d -1}{q-1}$ by $d$ and $t$ is the reminder from the division of $q^d$
by $d+1$.
\end{theorem}

The additional $\ell$ in the last summation in Theorem~\ref{thm:first_upper_bound} is required in case that each of the $\ell$ leftovers
from~$T_d$ can be combined with leftovers from $T$ to form a recovery set of size $d+1$.
The upper bound of Theorem~\ref{thm:first_upper_bound} can be further improved, but this will be left for future research.

\section{Conclusion and Problems for Future Research}
\label{sec:conclusion}

We have looked for the maximum number of recovery sets that can be obtained for any given $d$-subspace of $\F_q^k$
when the stored elements are all the 1-subspaces of $\F_q^k$, i.e., the columns of the generator matrix of
the $[(q^k -1)/(q-1),k,q^{k-1}]$ simplex code, which are also the points of the projective geometry PG$(k-1,q)$.
Lower and upper bounds on the number of recovery sets are provided, some of which are tight. Similar bounds can be given using
similar techniques. For example, the following bound was obtained and it will be given without details due to its length.

\begin{theorem}
\label{thm:bounds_d=6}
For $k \geq 7$,
$$
\left\lfloor \frac{91 \cdot 2^{k-6} +12}{10}  \right\rfloor \leq N_2(k,6)
\leq \left\lfloor \frac{91 \cdot 2^{k-6} +35}{10} \right\rfloor .
$$
\end{theorem}

The main open problems are associated with partitions of the leftovers into recovery sets. There
are several problems whose solutions will yield an optimal or almost optimal number of recovery sets:
\begin{itemize}
\item For $d \leq q-1$, find a partition of the 1-subspaces of $\F_q^n$ into $(d+2)$-subsets and possible one subset of
a smaller size, such that each $(d+2)$-subset spans a 2-subspace.

\item When $d > q-1$ and $k$ large, what is the minimum number of leftovers from different rows of~$T$ which are required
to form one recovery set? Find a partition of the 1-subspaces of $\F_q^n$ into subsets of this minimum size, such that each one
can form a recovery set.

\item How many leftovers (not from $T_d$ or the first row of $T$ in the binary case) are required to form one $d$-subspace,
where there is a single leftover in each row?

\item What is the size of the recovery sets that are obtained from the leftovers when $q=2$ and ${2^d \equiv 1~(\text{mod}~d+1)}$, i.e.,
one leftover from each internal row of $T$?

\item Write a program to produce an automatic integer programming problem for the number of recovery sets.
The program should be able to develop inequalities as in Section~\ref{sec:Integer_program} and as analyzed in
Section~\ref{sec:upper_binary}. Use the program to improve the upper bounds.
\end{itemize}

%
%
%
%
%

%

\bibliographystyle{IEEEtran}

\begin{thebibliography}{18}
\setlength{\emergencystretch}{0pt}
\hfuzz = 3.75pt

\providecommand{\url}[1]{#1}
\csname url@rmstyle\endcsname

\bibitem{AsYa18}
    {\sc H. Asi and E. Yaakobi,}
    {\sl Nearly optimal constructions of PIR and batch codes,}
    {\em IEEE Trans. Inform. Theory}, vol. 65, pp.\,947--964, 2019.

\bibitem{Beu75}
    {\sc A. Beutelspacher,}
    {\sl Partial spreads in finite projective spaces and partial designs,}
    {\em Math. Z.}, vol.\,145, pp.\,211--229, 1975.

\bibitem{BlEt19}
    {\sc S. R. Blackburn and T. Etzion,}
    {\sl PIR array codes with optimal virtual server rate,}
    {\em IEEE Trans. Inform. Theory}, vol.\,65, pp.\,6136--6145, 2019.

\bibitem{BEOVW}
     {\sc M. Braun, T. Etzion, P. R. J. \"Osterg\aa rd, A. Vardy, and A. Wassermann,}
     {\sl Existence of $q$-analogs of Steiner systems,}
     {\em Forum of Mathematics, Pi}, vol. 4, pp. 1--14, 2016.

\bibitem{CEKZ20}
    {\sc Y. M. Chee, T. Etzion, H. M. Kiah, and H. Zhang,}
    {\sl Recovery sets for subspaces from a vector space,}
    {\em IEEE Int. Symp. on Inf. Theory (ISIT)}, pp.\,542--547, Los Angeles, Jun. 2020.

\bibitem{CKYZ19}
    {\sc Y. M. Chee, H. M. Kiah, A. Yaakobi, and H. Zhang,}
    {\sl A generalization of the Blackburn-Etzion construction for private information retrieval array codes,}
    {\em IEEE Int. Symp. on Inf. Theory (ISIT)}, pp.\,1062--1066, Paris, Jul. 2019.

\bibitem{Chv83}
    {\sc V. Chvatal,}
    {\em Linear Programming}, Macmillan, 1983.

\bibitem{Del78}
       {\sc P. Delsarte,}
       {\sl Bilinear forms over a finite field, with applications to coding theory,}
       {\em J. of Combin. Theory, Series A}, vol. 25, pp. 226--241,~1978.

\bibitem{Etz22}
    {\sc T. Etzion,}
    {\em Perfect Codes and Related Structues}, World Scientific, 2022.

\bibitem{EtSi09}
    {\sc T. Etzion and N. Silberstein,}
    {\sl Error-correcting codes in projective
         space via rank-metric codes and Ferrers diagrams,}
    {\em IEEE Trans. Inform. Theory,} vol.\,55, pp.\,2909--2919, 2009.

\bibitem{EtSi13}
  {\sc T. Etzion and N. Silberstein,}
  {\sl Codes and designs related to lifted MRD codes,}
   {\em IEEE Trans.\ Inform. Theory}, vol.\,59, pp.\,1004--1017, 2013.

\bibitem{EtSt16}
    {\sc T. Etzion and L. Storme,}
    {\sl Galois geormetries and coding theory,}
    {\em Designs, Codes, and Crypt.,} vol.\,78, pp.\,311--350, 2016.

\bibitem{EtVa11}
    {\sc T.\,Etzion and A. Vardy,}
    {\sl Error-correcting codes in projective space,}
    {\em IEEE Trans.\ Inform.\ Theory}, vol.\,57, no.\,2, pp.\,754--763, 2011.

\bibitem{FVY15}
     {\sc A. Fazeli, A. Vardy, and E. Yaakobi,}
     {\sl Private information retrieval without storage overhead: coding instead of replication,}
     {\em arxiv.org/abs/1505.0624}, May 2015.

\bibitem{FVY15a}
    {\sc A. Fazeli, A. Vardy, and E. Yaakobi,}
    {\sl Coded for distributed PIR with low storage overhead,}
    {\em IEEE Int. Symp. on Inf. Theory (ISIT)}, pp.\,2852--2856, Hong Kong, Jun. 2015.

\bibitem{Gab85}
      {\sc E. M. Gabidulin,}
      {\sl Theory of codes with maximum rank distance,}
      {\em Problems of Information Transmission}, vol.~21, pp. 1--12, 1985.

\bibitem{HKRS}
     {\sc H. D. Hollmann, K. Khathuria, A. E. Riet, and, V. Skachek,}
     {\sl On some batch code properties of the simplex code,}
     {\em Designs, Codes, and Crypt.}, vol.\,91, pp.\,1595--1605, 2023.

\bibitem{HoPa72}
    {\sc S. J. Hong and A. M. Patel,}
    {\sl A general class of maximal codes for computer applications,}
    {\em IEEE Trans. on Comput.}, vol.~21, pp. 1322--1331, 1972.

\bibitem{KoKu08}
    {\sc A.\,Kohnert and S.\,Kurz,}
    {\sl Construction of large constant-dimension~codes with a prescribed minimum distance,}
    {\em Lecture Notes in Computer Science}, vol.\,5393, pp.\,31--42, 2008.

\bibitem{KoKs08}
    {\sc R. K\"{o}tter and F. R. Kschischang,}
    {\sl Coding for errors and erasures in random network coding,}
    {\em IEEE Trans. on Inform. Theory}, vol.~54, pp. 3579--3591, 2008.

\bibitem{KuYa21}
    {\sc S. Kurz and E. Yaakobi,}
    {\sl PIR codes with short block length,}
    {\em Designs, Codes, and Crypt.}, vol.\,89, pp.\,559--587, 2021.

\bibitem{NaYa22}
    {\sc M. Nassar and E. Yaakobi,}
    {\sl Array codes for functional PIR and batch codes,}
    {\em IEEE Trans. Inform. Theory}, vol.\,68, pp.\,839--862, 2022.

\bibitem{RaEt15}
    {\sc N. Raviv and T. Etzion,}
    {\sl Distributed storage systems based on intersecting subspace codes,}
    Proc. {\em International Symposium on Information Theory}, Hong Kong, pp.\,1462--1466, 2015.

\bibitem{RPDV16}
    {\sc A. S. Rawat, D. S. Papailiopoulos, A. G. Dimakis, and S. Vishwanath,}
    {\sl Locality and availability in distributed storage,}
    {\em IEEE Trans. Inform. Theory}, vol. 62, no. 8, pp.\,4481--4493, 2016.

\bibitem{Rot91}
    {\sc R. M. Roth,}
    {\sl Maximum-rank array codes and their application to crisscross error correction,}
    {\em IEEE Trans.\ Inform. Theory}, vol.~37, pp. 328--336, 1991.

\bibitem{ScEt02}
       {\sc M. Schwartz and T. Etzion,}
       {\sl Codes and anticodes in the Grassmann graph,}
       {\em J. of Combin. Theory, Series A}, vol. 97, pp. 27--42,~2002.

\bibitem{Seg64}
    {\sc B. Segre,}
    {\sl Teoria di Galois, fibrazioni proiettive e geometrie non desarguesiane,}
    {\em Ann. Mat. Pura Appl.}, vol.\,64, pp.\,1--76, 1964.

\bibitem{SES19}
    {\sc N. Silberstein, T.\,Etzion and M. Schwartz,}
    {\sl Locality and availability of array codes constructed from subspaces,}
    {\em IEEE Trans.\ Inform.\ Theory}, vol.\,65, pp.\,2648--2660, 2019.

\bibitem{SKK08}
    {\sc D. Silva, F.\,R.\ Kschischang, and R.\ Koetter,}
    {\sl A rank-metric approach to error control in random network coding,}
    {\em IEEE Trans. on Inform. Theory}, vol.\,54, pp.\,3951--3967, 2008.

\bibitem{VaYa23}
    {\sc A. Vardy and E. Yaakobi,}
    {\sl Private information retrieval without storage overhead: Coding instead of replication,}
    {\em IEEE J. Sel. Areas Inf. Theory}, vol.\,4, pp.\,286--301, 2023.

\bibitem{WKCB17}
    {\sc Z. Wang, H. M. Kiah, Y. Cassuto and J. Bruck,}
    {\sl Switch codes: codes for fully parallel reconstruction,}
    {\em IEEE Trans. Inform. Theory}, vol.\,63, pp.\,2061--2075, 2017.

\bibitem{YoYa22}
    {\sc L. Yohananov and E. Yaakobi,}
    {\sl Almost optimal construction of functional batch codes using extended simplex code,}
    {\em IEEE Trans. Inform. Theory}, vol.\,68, pp.\,6434--6451, 2022.

\bibitem{ZEY19}
    {\sc Y. Zhang, T. Etzion, and E. Yaakobi,}
    {\sl Bounds on the length of functional PIR and batch codes,}
    {\em IEEE Trans. Inform. Theory}, vol.\,66, pp.\,4917--4934, 2020.


\end{thebibliography}

\end{document}